\documentclass[11pt,a4paper, oneside]{memoir}
\usepackage{amssymb}
\usepackage{amsthm}
\usepackage{mathtools}		
\usepackage{wrapfig}
\usepackage{geometry}
\usepackage[utf8]{inputenc} 	
\usepackage{bbm} 			
\usepackage{enumitem} 		
\usepackage{cleveref}		
\usepackage{etoolbox}    	
\usepackage{csquotes}		
\usepackage{authblk}

\usepackage[style=ext-numeric,articlein=false,maxbibnames=10,backend=biber,sorting=none,isbn=false,url=false,doi=false,giveninits=true]{biblatex}
\AtEveryBibitem{%
  \clearlist{language}%
  \clearfield{number}%
}
\DeclareFieldFormat[article,periodical]{volume}{\mkbibbold{#1}}
\DeclareDelimFormat[bib,biblist]{nametitledelim}{.\space}
\DeclareDelimFormat[bib,biblist]{jourvoldelim}{,\space}

\newcommand{\irrep}{\mathcal{H}}
\DeclareMathOperator{\identity}{\mathbbm{1}}
\DeclareMathOperator{\Op}{Op}
\DeclareMathOperator{\Hus}{H}

\newcommand{\cQ}{\mathcal{Q}}
\newcommand{\tensor}{\otimes}

\newcommand{\Z}{\mathbb{Z}}
\newcommand{\N}{\mathbb{N}}
\newcommand{\C}{\mathbb{C}}

\DeclarePairedDelimiter\angledbrackets\langle\rangle
\DeclarePairedDelimiter\curlybrackets\lbrace\rbrace
\DeclarePairedDelimiter\verticalline\lvert\rvert
\DeclarePairedDelimiter\doubleverticalline\lVert\rVert
\DeclarePairedDelimiter\bra{\langle}{\rvert} 
\DeclarePairedDelimiter\ket{\rvert}{\rangle}

\newcommand{\abs}[1]{\verticalline*{#1}}
\newcommand{\iabs}[1]{\verticalline{#1}}

\newcommand{\iset}[1]{\curlybrackets{#1}} 
\newcommand{\norm}[1]{\doubleverticalline*{\ifblank{#1}{{}\cdot{}}{#1}}}
\newcommand{\inorm}[1]{\doubleverticalline{\ifblank{#1}{{}\cdot{}}{#1}}}

\newcommand\poisson[2]{\curlybrackets*{\ifblank{#1}{{}\cdot{}}{#1},\ifblank{#2}{{}\cdot{}}{#2}}}

\newcommand\braket[2]{\angledbrackets*{\ifblank{#1}{{}\cdot{}}{#1}\vert\ifblank{#2}{{}\cdot{}}{#2}}}
\newcommand{\ketbra}[3]{\ket{\ifblank{#1}{{}\cdot{}}{#1}}_{\ifblank{#3}{}{#3\,#3}}\bra{\ifblank{#1}{{}\cdot{}}{#2}}}

\DeclareMathOperator{\dx}{\,d\mathnormal{x}}
\DeclareMathOperator{\dt}{\,d\mathnormal{t}}
\DeclareMathOperator{\domega}{\,d\mathnormal{\omega}}
\DeclareMathOperator{\dg}{\,d\mathnormal{g}}
\newcommand{\dd}{\,\mathrm{d}}

\DeclareMathOperator{\SU}{SU}

\DeclareMathOperator{\Tr}{Tr}

\theoremstyle{plain}
\newtheorem{definition}{Definition}[chapter] 
\newtheorem{theorem}[definition]{Theorem} 
\newtheorem{lemma}[definition]{Lemma}
\newtheorem{proposition}[definition]{Proposition}
\newtheorem{corollary}[definition]{Corollary}

\theoremstyle{definition}
\newtheorem*{question}{Question}

\usepackage{indentfirst}
\counterwithout{section}{chapter}
\counterwithout{definition}{chapter}
\numberwithin{equation}{section}
\setsecnumdepth{subsection}

\DeclareMathOperator{\image}{Im}

\title{ SU(2)-equivariant quantum channels:\\ semiclassical analysis}
\author[1]{Tommaso Aschieri}
\author[1]{B\l a\.zej Ruba}
\author[1]{Jan Philip Solovej}
\affil[1]{Department of Mathematical Sciences, University of Copenhagen, \protect\\
Universitetsparken 5, 2100 Copenhagen, Denmark}
\date{\today}

\begin{document}

\maketitle

\begin{abstract}
We study completely positive and trace-preserving equivariant maps between operators on irreducible representations of $\SU(2)$. We find asymptotic approximations of channels in the limit of large output representation and we compute traces of functions of channel outputs. Our main tool is quantization using coherent states. We provide quantitative error bounds for various semiclassical formulas satisfied by quantizations of functions on the sphere. 
\end{abstract}

\section{Introduction}

We study equivariant (also called covariant or invariant) quantum channels between irreducible representations of $\SU(2)$. That is, we consider completely positive, trace-preserving and equivariant maps $ \Phi : B(\irrep_J) \to B(\irrep_K)$, where $\irrep_J, \irrep_K$ are the irreducible representations of $\SU(2)$ labeled by the spin parameters $J,K \in \{ 0,  \frac12, 1 , \dots \}$, and $B(\irrep_J)$ is the algebra of linear operators on $\irrep_J$. 

Equivariant quantum channels $B(\irrep_J) \to B(\irrep_K)$ form a simplex \cite{nuwairan_su2-irreducibly_2013} with vertices $\Phi_{J,K}^M$, $M \in \{ |K-J|, \dots, K+J \}$ (see \eqref{eq:channel_formula_1} below for the definition of $\Phi_{J,K}^M$). If $K$ is sufficiently large, we approximate $\Phi_{J,K}^{K+i}$ by a composition of a \enquote{symbol} map $\Hus_J^{-i} : B(\irrep_J) \to C^\infty(S^2)$ and a \enquote{quantization} map $\Op_K : L^1(S^2) \to B(\irrep_K)$. That is, we approximate
\begin{equation}
    \Phi_{J,K}^{K+i} \approx \frac{2J+1}{2K+1} \Op_K \Hus_J^{-i}.
    \label{eq:channel_approx}
\end{equation}
Note the factorization of the dependence on $K$ and on $J$, $i$.

The map $\Op_K$ in \eqref{eq:channel_approx} is a quantization map, with $\frac{1}{K}$ playing the role of a semiclassical parameter, in the following sense. Firstly, for regular enough functions $f,g$ on $S^2$ and large $K$ one has:
\begin{equation}
    \Op_K(f) \Op_K(g) \approx \Op_K(fg) .
\label{eq:multiplicativity_vague}
\end{equation}
Furthermore, the commutator of $\Op_K(f)$ and $\Op_K(g)$ is given, up to terms vanishing faster for $K \to \infty$, in terms of the Poisson bracket \eqref{eq:poisson_bracket} of $f$ and $g$:
\begin{equation}
[\Op_K(f),\Op_K(g)] \approx \frac{i}{K} \Op_K(\poisson{f}{g}) .
\label{eq:com_bracket}
\end{equation}
Let us also mention the slightly stronger statement
\begin{equation}
\Op_K(f) \Op_K(g) \approx \Op_K \left( fg + \frac{i \{ f , g \} - \nabla f \cdot \nabla g }{2K} + O \left( \frac{1}{K^2} \right)  \right).
\label{eq:product_vague}
\end{equation}
Formulas \eqref{eq:multiplicativity_vague}, \eqref{eq:com_bracket} and \eqref{eq:product_vague} are stated more precisely in \Cref{thm:multiplicativity} below.

One advantage of approximating outputs of channels with operators of the form $\Op_K(f)$ is that we can approximate the functional calculus,
\begin{equation}
    \varphi(\Op_K(f)) \approx \Op_K (\varphi \circ f),
\end{equation}
and we have the semiclassical trace formula:
\begin{equation}
\frac{1}{2K+1} \Tr \left[ \varphi (\Op_K(f)) \right] \approx \int_{S^2} \varphi(f(\omega)) \domega, 
\label{eq:semiclassical_trace}
\end{equation}
where $\domega$ is the normalized area element on the sphere, $\int_{S^2} \domega =1$. See \Cref{thm:traces} for a precise statement of \eqref{eq:semiclassical_trace}. The formula \eqref{eq:semiclassical_trace} can be used for example to compute entropies of channel outputs, which we detail in \Cref{thm:output_entropies}.

We need some preparations before we state our results precisely. Let $\ket{\omega;i}_J \in \irrep_J$ be a~normalized eigenvector of spin in the direction $\omega \in S^2$ with eigenvalue $i \in \{ -J, \dots, J \}$:
\begin{align}
    \omega\cdot \vec S \ket{\omega;i }_J = i \ket{\omega;i}_J,
\end{align}
where $\vec S = (S_x, S_y, S_z)$ are the spin operators on $\irrep_J$, see \Cref{sec:prep}. If $i = J$, we often hide the index $i$. States $\ket{\omega}_J$ are often called (Bloch or spin) coherent states. Famously, coherent states allow to construct a resolution of the identity:
\begin{equation}
    (2J+1) \int_{S^2}\ketbra{\omega}{\omega}{J} \domega =\identity,
    \label{eq:coherent_resolution_of_1}
\end{equation}
Equation \eqref{eq:coherent_resolution_of_1} remains true if we replace $\ket{\omega}_J {}_J\bra{\omega}$ by $\ket{\omega;i}_J {}_J \bra{\omega;i}$.

Using coherent states one defines the Husimi function of an operator $\rho \in B(\irrep_J)$:
\begin{equation}
    \Hus_J (\rho)(\omega) = {}_J \bra{\omega} \rho \ket{\omega}_J,
\end{equation}
and the quantization of a function $f$ on $S^2$:
\begin{equation}
    \Op_J (f) = (2J+1) \int_{S^2} f(\omega) \ketbra{\omega}{\omega}{J} \domega.
\end{equation}
We also consider generalizations of $\Hus_J$, $\Op_J$ constructed using $\ket{\omega;i}_J$:
\begin{gather}
    \Hus_J^i (\rho)(\omega)  = {}_J \bra{\omega ; i } \rho \ket{\omega ; i}_J, \\
    \Op_J^i (f)  = (2J+1) \int_{S^2} f(\omega) \ketbra{\omega;i}{\omega;i}{J} \domega.
\end{gather}

Let us mention some elementary properties of $\Hus_J^i$ and $\Op_J^i$. All these maps are positivity preserving, unital, and they preserve the \enquote{total mass}, understood as the integral over $S^2$ for functions and as the normalized trace for operators. The map $(2J+1) \Hus^i_J$ is the adjoint of $\Op^i_J$.

The maps $(2J+1) \Hus^i_J$ are the vertices of the simplex of $\SU(2)$-equivariant linear maps taking density matrices in $B(\irrep_J)$ to probability measures\footnote{Here and in some other parts of the manuscript we do not distinguish between an integrable function $f$ and an absolutely continuous measure $f(\omega) \domega$.} on $S^2$, see \Cref{prop:qtc_simplex} in \Cref{sec:quantization}, where we also discuss the equivalence between these maps and equivariant POVMs (positive operator valued measures). 

The maps $\Op_J, \Hus_J$ appear in the work of Berezin \cite{berezin_general_1975}. Berezin calls $\Hus_J(\rho) $ the covariant symbol of $\rho$, and any function $f$ such that $\rho= \Op_J(f)$ a contravariant symbol of $\rho$. The paper \cite{berezin_general_1975} includes a proof that every operator admits a contravariant symbol (obviously not unique), and discusses interesting inequalities satisfied by $\rho$ given in terms of its symbols. If $\rho = \Op_J(f)$ and $g = \Hus_J(\rho)$, we have the norm bounds:
\begin{equation}
\| g \|_\infty \leq \| \rho \|_\infty \leq \| f \|_\infty,
\label{eq:norm_est_lower_upper}
\end{equation}
and, in the case of real $f$ (hence self-adjoint $\rho$) and for a convex function $\varphi$ defined on $\image(f)$, the Berezin-Lieb inequalities
\begin{equation}
\int_{S^2} \varphi(g(\omega)) \domega \leq \frac{1}{2J+1} \Tr[\varphi(\rho)] \leq \int_{S^2} \varphi(f(\omega)) \domega.
    \label{eq:Berezin-Lieb}
\end{equation}
The inequalities \eqref{eq:Berezin-Lieb} appear in \cite{berezin_general_1975}. Lieb \cite{lieb_classical_1973} used them with $\varphi(x)= e^x$ to compare partition functions of quantum lattice spin models to partition functions of corresponding classical systems, with enough precision to make statements about the thermodynamic limit. The inequalities \eqref{eq:Berezin-Lieb} were reproved (in a slightly more general setting) by Simon in \cite{simon_classical_1980}, where another proof of existence of contravariant symbols was also given. Both \eqref{eq:norm_est_lower_upper} and \eqref{eq:Berezin-Lieb} contain $g$ in the lower bound and $f$ in the upper bound, and for this reason Simon \cite{simon_classical_1980} proposed to call $g$ the lower symbol of $\rho$ and $f$ an upper symbol of $\rho$. 

Quantization using Bloch coherent states is essentially the same as Berezin-Toeplitz quantization applied to the manifold $S^2$. We refer the reader to the review article \cite{schlichenmaierBerezinToeplitzQuantizationCompact2010} and the book \cite{le2018brief} for pedagogical introductions to the general framework of Berezin-Toeplitz quantization, summary of known results and explanation of relations with other approaches to quantization, such as geometric quantization or deformation quantization. In order to facilitate comparisons of results, we briefly discuss the relation between Toeplitz operators and the operators $\Op_J(f)$ that we use in \Cref{app:BT}.

The use of Bloch coherent states in the representation theory of $\SU(2)$ is largely analogous to the use of Glauber coherent states to analyze representations of canonical commutation relations. A popular alternative to quantization with Glauber coherent states is Weyl quantization. It has many advantages, e.g.\ it is equivariant with respect to a larger symmetry group and it is characterized by smaller error terms in certain (approximate) semiclassical formulas. In contrast to quantization using coherent states, Weyl quantization does not preserve positivity: the quantization of a~positive function is not necessarily a positive operator. An~analog of Weyl quantization for representations of $\SU(2)$ goes back to Stratonovich \cite{stratonovich} and has been further developed by V{\'a}rilly and Gracia-Bond{\'i}a \cite{varilly_moyal_1989}. Let us summarize this briefly. A Stratonovich-Weyl correspondence is an SU(2)-equivariant linear map $\sigma_J : B(\irrep_J) \to L^2(S^2) $ such that $\sigma_J(\rho^*) = \overline{\sigma_J(\rho)}$, $\sigma_J(\mathbbm 1) = 1$ and 
\begin{equation}
    \frac{1}{2J+1} \Tr \left[ \rho_1^* \rho_2 \right] = \int_{S^2} \overline{\sigma_J(\rho_1) (\omega)} \sigma_J(\rho_2)(\omega) \domega.
\label{eq:Stratonovich_Weyl_isometry}
\end{equation}
Given one such $\sigma_J$, one introduces a quantization map $\Op_J^{\mathrm{SW}}$ defined as the adjoint of $(2J+1)\sigma_J$. The property \eqref{eq:Stratonovich_Weyl_isometry} is equivalent to $\Op_J^{\mathrm{SW}}$ being a left inverse of $\sigma_J$. For every $J$ there exist $2^{2J}$ Stratonovich-Weyl correspondences $\sigma_J$. In \cite{varilly_moyal_1989} a proposal of a~distinguished Stratonovich-Weyl correspondence was given. It was shown \cite{varilly_moyal_1989} that the distinguished $\sigma_J$ can be described with the formula
\begin{equation}
    \sigma_J = \Hus_J (\Op_J \Hus_J)^{- \frac12}, \qquad  \Op_J^{\mathrm{SW}} = (\Op_J \Hus_J)^{- \frac12} \Op_J.
\end{equation}
In this paper we prefer to use directly the map $\Hus_J$, which is not a Stratonovich-Weyl correspondence because $\Op_J \Hus_J \neq 1$. We remark that $\Op_J$ is an optimal positivity preserving quantizaton, in the sense explained in \cite{ioosBerezinToeplitzQuantization2021}.

The map $\Hus_J \Op_J$ is often called the Berezin transform. It is an approximate identity, in the sense that $\Hus_J \Op_J \to 1$ strongly as $J \to \infty$ \cite{berezin_general_1975}. Here we describe this convergence quantitatively. We bound the operators
\begin{equation*}
 (1-\Hus_J \Op_J)(1-4 \Delta)^{-1} \text{ and } (1-\Op_J \Hus_J )(1+4 \mathcal Q)^{-1},  
\end{equation*}
where $\Delta$ is the Laplacian on $S^2$ and $\mathcal Q$ is the quadratic Casimir of SU(2) acting on $B(\irrep_J)$ by iterated commutators, see \eqref{eq:ad_Casimir}. This allows to estimate the difference of the upper and lower symbols of an operator given quantified regularity of the upper symbol, and to control the error in the approximation of an operator $\rho \in B(\irrep_J)$ by $\Op_J \Hus_J (\rho)$, given bounds on $\rho$ and $\mathcal Q (\rho)$. In \Cref{lemma:HJ_OpJ_composition,lemma:HJ_OpJ_as_integral_operator} we give explicit formulas for $\Op_J \Hus_J$ and $\Hus_J \Op_J$ which are useful in deriving the estimates just described. See also \cite[Equation 5.9]{berezin_general_1975} for an interesting infinite product representation. 

We now summarize the main results of the paper. In the bounds below, $\| \cdot \|_p$ is the $p$th Schatten norm of an operator or the $L^p$ norm of a~function, depending on the context.

\begin{theorem} \label{thm:hus_op}
$\iset{\Hus_J \Op_J}_{J}$ is a monotonically increasing sequence of positive, positivity-preserving operators. It converges strongly to $\identity$ on $L^p(S^2)$ for every $p \in [1,\infty)$ and on $C(S^2)$. For~every $s \in [0,1]$ we~have the bounds
\begin{align}
    \| (1-\Hus_J \Op_J ) f \|_2 & \leq \frac{1}{(2J+1)^s} \| (-\Delta)^s f \|_2,  \\
    \| (1- \Op_J \Hus_J) \rho \|_2 & \leq \frac{1}{(2J+1)^s} \| \cQ^s \rho \|_2,
\end{align}
and
\begin{align}
        \norm{\left(1-\Hus_J \Op_J+\frac{\Delta}{2J+1}\right)f}_{2} & \leq\frac{1}{(2J+1)^{1+s}} \| 
(-\Delta)^{1+s} f\|_2,  \\
\norm{\left(1- \Op_J \Hus_J-\frac{\cQ}{2J+1}\right)\rho}_{2} 
& \leq\frac{1}{(2J+1)^{1+s}} \| 
\cQ^{1+s} \rho \|_2.
    \end{align}
Moreover, there exists a constant $C>0$ such that for every $p \in [1 , \infty]$ we have
\begin{align}
    \| (1-\Hus_J \Op_J ) f \|_p & \leq \frac{C}{2J+1} \| (1-4\Delta) f \|_p,  \\
    \| (1- \Op_J \Hus_J) \rho \|_p & \leq \frac{C}{2J+1} \| (1+4\cQ) \rho \|_p.
    \end{align}
\end{theorem}

\Cref{thm:hus_op} is proved in \Cref{sec:inversion}.

Next, we describe the approximations \eqref{eq:multiplicativity_vague}, \eqref{eq:com_bracket} and \eqref{eq:product_vague} with precisely stated error bounds. In what follows, $\| f \|_{k,p}$ is the Sobolev norm of a function $f \in W^{k,p}(S^2)$ and $\| f \|_{C^{0,\alpha}_h}$ is the H\"older seminorm; see \eqref{eq:Sobolev_norm} and \eqref{eq:Holder_seminorm} for definitions. We also sometimes use the notation $\| \cdot \|_{C^{0, \alpha}_h}$ for the standard H\"older seminorm of a function defined on a (possibly unbounded) interval. 

We note that under mild regularity assumption on $f$, the norm $\| \Op_J(f) \|_p$ grows with $J$ with the rate $(2J+1)^{\frac{1}{p}}$ (see \Cref{lem:oph_norms} and \Cref{thm:traces}). For this reason we formulate some bounds, e.g. \eqref{eq:mult_bound_thm} below, with a negative power of $J$ included on the left hand side of the inequality. 

\begin{theorem} \label{thm:multiplicativity}
There exists a constant $C>0$ such that the following bounds are true.

If $p,p_1,p_2 \in [1, \infty]$ satisfy $\frac{1}{p}= \frac{1}{p_1 } + \frac{1}{p_2}$ and $f \in W^{2,p_1}(S^2)$, $g \in W^{2,p_2}(S^2)$, then:
\begin{align}
 \| \Hus_J(\Op_J(f) \Op_J(g))  - fg \|_p & \leq \frac{C}{2J+1} \| f \|_{2,p_1} \| g \|_{2,p_2}, \\
(2J+1)^{- \frac{1}{p}} \| \Op_J(f) \Op_J(g) - \Op_J(fg) \|_p  & \leq   \frac{C}{2J+1} \| f \|_{2,p_1} \| g \|_{2,p_2}. \label{eq:mult_bound_thm}
\end{align}

If $f$ is $\alpha$-H\"older continuous and $g $ is $\beta$-H\"older continuous, $\alpha, \beta \in [0,1]$, then
\begin{equation}
    \frac{1}{2J+1} \left \| \frac{\Op_J(f) \Op_J(g) + \Op_J(g) \Op_J(f)}{2} - \Op_J(fg) \right\|_1 \leq C\frac{ \| f \|_{C^{0, \alpha}_h} \| g \|_{C^{0, \beta}_h}}{(2J+1)^{\frac{\alpha + \beta}{2}}}.
    \label{eq:Holder_bound_thm}
\end{equation}

If $f, g \in W^{4,2}(S^2)$, then
\begin{align}
\left \| \Hus_J(\Op_J(f) \Op_J(g)) - \left( fg + \frac{i \poisson{f}{g} - \nabla f \cdot \nabla g + \Delta(fg)}{2J+1} \right) \right \|_1 & \leq C\frac{\| f \|_{4,2} \| g \|_{4,2}}{(2J+1)^2}, \\
\frac{1}{2J+1} \left \| \Op_J(f) \Op_J(g) - \Op_J \left( fg + \frac{i \poisson{f}{g} - \nabla f \cdot \nabla g }{2J+1} \right) \right \|_1 & \leq C\frac{\| f \|_{4,2} \| g \|_{4,2}}{(2J+1)^2} . \label{eq:Op_prod_expansion}
\end{align}
\end{theorem}

\Cref{thm:multiplicativity} is proved in \Cref{sec:mult}. We remark that results similar to \Cref{thm:multiplicativity} are known quite generally for Berezin-Toeplitz operators \cite{bordemannToeplitzQuantizationKahler1994}, also including higher order terms. In fact, the composition $\Op_J(f) \Op_J(g)$ admits an asymptotic expansion for $J \to \infty$ given by $\Op_J(f \star_{\mathrm{BT}} g)$, where $\star_{\mathrm{BT}}$ is a star product on $S^2$ \cite{schlichenmaierZweiAnwendungenAlgebraischgeometrischer1996,schlichenmaierDeformationQuantizationCompact2000}. In the first results of this type remainder terms contained constants depending in an unspecified way on functions $f,g$, which were assumed to be smooth. It has later been established \cite{barronSemiclassicalPropertiesBerezin2014,charlesSharpCorrespondencePrinciple2018} that the error terms can be controlled by $C^k$ norms of $f,g$ (with large enough explicit~$k$). Remainders in our \Cref{thm:multiplicativity} are bounded by Sobolev norms, so they require less regularity of $f,g$. We believe that one could extend \Cref{thm:multiplicativity} to include higher order terms by using the same proof method, but we did not pursue this.

Let us remark that terms of order $\frac{1}{2J+1}$ in \eqref{eq:Op_prod_expansion} involve only one derivative of $f$ and $g$, so one could hope for the error bound in \eqref{eq:mult_bound_thm} to depend on only one, not two derivatives of $f$ and $g$. We ask the following:

\begin{question}
Does there exist a constant $C>0$ such that, for $p, p_1, p_2$ as in \eqref{eq:mult_bound_thm}:
\begin{equation}
(2J+1)^{- \frac{1}{p}} \| \Op_J(f) \Op_J(g) - \Op_J(fg) \|_p \leq \frac{C}{2J+1} \| f \|_{1,p_1} \| g \|_{1,p_2}?
\end{equation}

\end{question}

Our bound \eqref{eq:Holder_bound_thm} is significantly weaker, but (putting $\alpha = \beta = 1$) it does depend only on first derivatives of $f$ and $g$. Similarly, one could ask if an inequality analogous to \eqref{eq:Op_prod_expansion} holds for $f,g \in W^{2,2}(S^2)$.

By a standard density argument, \Cref{thm:multiplicativity} implies that if $p,p_1,p_2 \in [1, \infty]$ satisfy $1/p = 1/p_1 + 1/p_2$, then
\begin{equation}
\lim_{J \to \infty}   (2J+1)^{- \frac{1}{p}} \| \Op_J(f) \Op_J(g) - \Op_J(fg) \|_p = 0 
\end{equation}
for all $f \in L^{p_1}(S^2)$ and $g \in L^{p_2}(S^2)$ (if $p_i = \infty$, $L^{p_i}(S^2)$ has to be replaced by $C(S^2)$).

We prove also the semiclassical formula \eqref{eq:semiclassical_trace} for traces. This result consists of several statements with differing assumptions about the functions $\varphi$ and $f$.

\begin{theorem} \label{thm:traces}
We have the expansion:
\begin{equation}
    \frac{1}{2J+1}\Tr_J[\varphi(\Op_J(f))] = \int_{S^2}\varphi(f(\omega))\domega - \;\mathcal{E},
\end{equation}
where $\mathcal{E}$ is an error term that satisfies the following bounds.
\begin{enumerate}
    \item If $f\in W^{1,2}(S^2)$ is valued in a (possibly unbounded) interval $I$, and $\varphi$ is a function on $I$ with second derivative in $L^\infty(I)$, then $\abs{\mathcal{E}}\leq \frac{1}{2J+1}\norm{\varphi''}_\infty \norm{\nabla f}_{2}^2$.
    \item If $f\in W^{2,1}(S^2)$ is real-valued, and $\varphi\in C^{0,\alpha}(\image(f))$ is convex, $0<\alpha\leq 1$, then $0\leq {\mathcal{E}}\leq \frac{C}{(2J+1)^{\alpha}} \norm{\varphi}_{C_h^{0,\alpha}}\norm{(1-4\Delta)f}^\alpha_{1}$ with a universal constant $C$.
    \item If $f\in C(S^2)$ is real-valued and $\varphi\in C(\image(f))$, then $\mathcal{E}\to 0 $ as $J\to\infty$.
\end{enumerate}
\end{theorem}

We remark that \Cref{thm:traces} will be obtained by combining the Berezin-Lieb inequalities with \Cref{thm:hus_op}. Its proof is presented in \Cref{sec:mult}.

The function $\varphi(x) = x \log(x)$ is particularly important because it appears in the definition of the von Neumann entropy of a density matrix $\rho$:
\begin{align}
    S_{vN}(\rho) = -\Tr[\rho\log\rho].
\end{align}
In this case \Cref{thm:traces} can be used to obtain the following result:

\begin{corollary}
    Let $\varphi(x) = x\log(x)$ and let $f \geq 0$. Then
\begin{align}
    \frac{1}{2J+1}\Tr_J[\varphi(\Op_J(f))] = \int_{S^2}\varphi(f(\omega))\domega - \;\mathcal{E},
\end{align}
where $\mathcal E \geq 0$ is an error term satisfying the following bound. There exists $C>0$ such that for all $f \in W^{2,1}(S^2) \cap L^\infty(S^2)$ and $J\geq 1$ we have:
\begin{align}
        \mathcal{E} \leq & 
\;C \frac{\log(2J+1)}{2J+1}\norm{(1-4\Delta)f}^{1-\frac{1}{\log(2J+1)}}_{1} \\
        & +C \frac{1}{2J+1}\max\{ 0, 1+\log \| f \|_\infty \} \norm{(1-4\Delta)f}_{1}. \nonumber
\end{align}
\end{corollary}

As already mentioned, the set of equivariant channels $B(\irrep_J) \to B(\irrep_K)$ is a simplex with vertices $\{ \Phi_{J,K}^M \}$, $M \in \{ |K-J|, \dots, K+J \}$ (see also \Cref{Appendix:QC} for the case of a~general group $G$). The channel $\Phi_{J,K}^M$ can be defined in three equivalent ways:
\begin{align}
\Phi^M_{J,K}(\rho) &= \frac{2J+1}{2M+1} \Tr_J[P_{J,K}^M(\rho^\beta\tensor \identity_K)],\label{eq:channel_formula_1}
\shortintertext{and:}
\Phi^M_{J,K}(\rho) &= \frac{2J+1}{2K+1} q^K_{J,M}(\rho\tensor \identity_M)\iota^K_{J,M}, \label{eq:channel_formula_2}
\shortintertext{and also:}
\Phi^M_{J,K}(\rho) &= \Tr_{M}[\iota^J_{K,M}\,\rho\; q^J_{K,M}].
\label{eq:channel_formula_3}
\end{align}
Here $P^M_{J,K}$ is the orthogonal projection onto $\mathcal H_M$ in $\mathcal H_J \otimes \mathcal H_K$, $\rho^\beta$ is the adjoint of $\rho$ with respect to the invariant bilinear form on $\mathcal H_J$ \eqref{eq:definition_of_beta}, $\iota^K_{J,M}$ is an equivariant isometric embedding $\irrep_K \to \irrep_J \otimes \irrep_M$, and $q^K_{J,M}$ is its adjoint: $q^K_{J,M} = (\iota^K_{J,M})^*$. We remark that maps closely related to \eqref{eq:channel_formula_1} (with $\rho^\beta$ replaced by $\rho$) have been studied by Lieb and Solovej \cite{lieb_quantum_1991} under the name quantum coherent operators.

Given a collection $\lambda = \iset{\lambda_i}$ of non-negative numbers with sum $1$, we let:
\begin{align}
    \Phi_{J,K}^{(\lambda)} = \sum_{i} \lambda_i \Phi_{J,K}^{K+i}, \qquad \Hus_J^{( \lambda)} = \sum_{i} \lambda_i \Hus_J^{-i}.
\end{align}
Thus all equivariant channels $B(\irrep_J) \to B(\irrep_K)$ are of the form 
$\Phi_{J,K}^{(\lambda)}$, with unique $\lambda$. 

In some formulas we prefer to use the maps $\frac{2K+1}{2J+1}\Phi_{J,K}^{(\lambda)}$, which are unital. If $\rho$ is a~density matrix, $\frac{2K+1}{2J+1}\Phi_{J,K}^{(\lambda)}(\rho)$ has Schatten $p$-norm of order $\big( \frac{2K+1}{2J+1} \big)^{1/p}$.

We are now ready to state rigorously the approximation of equivariant channels announced earlier:
\begin{theorem}\label{thm:channel_approx}
    Let $\rho\in B(\irrep_J)$. Then for any $1\leq p\leq \infty$:
    \begin{align}
\left( \frac{2K+1}{2J+1} \right)^{- \frac{1}{p}} \norm{ \frac{2K+1}{2J+1} \Phi^{K+i}_{J,K}(\rho) - \Op_K \Hus_J^{-i} (\rho)  }_p &\leq 12 \frac{(J-i)(J+i+1)}{2K-J+i+1} \| \rho \|_p, \label{eq:channel_approx_i} \\
(2J+1)^{\frac{1}{p}}\norm{\frac{2K+1}{2J+1}\Hus_K\Phi^{K+i}_{J,K}(\rho)-\Hus^{-i}_J(\rho)}_{p}&\leq 2\frac{(J+i)(J-i+1)}{2K-J+i+1}\norm{\rho}_p. \label{eq:channel_hus_approx_i}
    \end{align}
    In particular, for $2J \leq K$:
    \begin{align}
        \left( \frac{2K+1}{2J+1} \right)^{- \frac{1}{p}} \norm{ \frac{2K+1}{2J+1} \Phi^{(\lambda)}_{J,K}(\rho) - \Op_K \Hus_J^{(\lambda)} (\rho)  }_p&\leq 6\frac{(2J+1)^{2}}{2K+1}  \norm{\rho}_p, \label{eq:channel_approx_general} \\
    (2J+1)^{\frac{1}{p}}\norm{\frac{2K+1}{2J+1}\Hus_K(\Phi^{(\lambda)}_{J,K}(\rho))-\Hus^{(\lambda)}_J(\rho)}_{p}&\leq \frac{(2J+1)^{2}}{2K+1}  \norm{\rho}_p.    \label{eq:channel_hus_approx_general}
    \end{align}
\end{theorem}

\Cref{thm:channel_approx} is proved in \Cref{sec:channels}.

Consider the channel $\Phi_{J,K}^{K+i}$ with large $K$, and possibly also large $J \ll K$. If we do not impose additional restrictions on $i$, \eqref{eq:channel_approx_i} gives an approximation of $\Phi_{J,K}^{K+i}$ which is accurate if $J^2 \ll K$. By the triangle inequality we have \eqref{eq:channel_approx_general}, which gives an approximation of an arbitrary channel $B(\irrep_J) \to B(\irrep_K)$ accurate if $J^2 \ll K$. On the other hand, if we fix $J-|i|$, the numerator on the right hand side of \eqref{eq:channel_approx_i} grows only linearly, not quadratically with $J$. Therefore, for channels with $|i| $ close to $J$ we have an accurate approximation under the condition $J \ll K$, much weaker than $J^2 \ll K$. Similar remarks apply to channels composed with the Husimi map, as seen in \eqref{eq:channel_hus_approx_i} and \eqref{eq:channel_hus_approx_general}.

Another interesting property of \eqref{eq:channel_approx_i} and \eqref{eq:channel_hus_approx_i} is that the error terms vanish for $J=i$ and $J=-i$, respectively. In these cases we have the exact identities:
\begin{align}
    \Phi_{J,K}^{K+J} & = \frac{2J+1}{2K+1} \Op_K \Hus_J, \\
    \Hus_K \Phi_{J,K}^{K-J} & = \frac{2J+1}{2K+1} \Hus_J^{-J}.
\end{align}

We remark that estimates similar to \eqref{eq:channel_approx_i} (for specific channels but with SU($d$) in place of SU(2)) have been used \cite{christandl_one-and--half_2007} to prove quantum de Finetti theorems. We hope that generalizations of \eqref{eq:channel_approx_i} will find other uses in many body quantum theory, e.g.\ in the $N$-representability problem.

We deduce the following for traces of functions of channel outputs:

\begin{theorem} \label{thm:output_traces}
If $\rho$ is a density matrix, then
\begin{equation}\label{eq:trace_formula_channels_with_error_term_intro}
    \frac{1}{2K+1}\Tr_K \left(\varphi \left(\frac{2K+1}{2J+1}\Phi_{J,K}^{K+i}(\rho)\right)\right) = \int_{S^2}\varphi(\Hus_J^{-i}(\rho))\domega + \mathcal E,
\end{equation}
where the error term $\mathcal E$ satisfies the following bounds:
\begin{enumerate}
\item If $\varphi$ has second derivative in $ L^\infty([0,1])$, then $|\mathcal E| \leq 10 \norm{\varphi''}_{\infty} \frac{J-\abs{i}+1}{2K-J+i+1}$;
\item If $\varphi \in C^{0, \alpha}([0,1])$ is convex, then $|\mathcal E| \leq C \norm{\varphi}_{C_{h}^{0,\alpha}} \left(\frac{J-\abs{i}+1}{2K-J+i+1}\right)^{\alpha}$ with a universal constant $C$.
\end{enumerate}

Similarly, if in \eqref{eq:trace_formula_channels_with_error_term_intro} we replace $\Phi_{J,K}^{K+i}$ with $\Phi_{J,K}^{(\lambda)}$ and $\Hus_J^{-i}$ with $\Hus_J^{(\lambda)}$, then under the assumption $K \geq 2J$ the error term $\mathcal E$ can be bounded as follows:
\begin{enumerate}
\item If $\varphi$ is a function with second derivative in $ L^\infty([0,1])$, then $|\mathcal E| \leq 4 \norm{\varphi''}_{\infty} \frac{2J+1}{2K+1}$;
\item If $\varphi \in C^{0, \alpha}([0,1])$ is convex, then $|\mathcal E| \leq C \norm{\varphi}_{C_{h}^{0,\alpha}} \left(\frac{2J+1}{2K+1}\right)^{\alpha}$ with a universal constant $C$;
\item If $\varphi \in C([0,1])$, then $\mathcal{E}\to 0$ as $J\to\infty$.
\end{enumerate}
\end{theorem}

\Cref{thm:output_traces} is proved in \Cref{sec:channels}. Here is its application to output entropies of quantum channels:

\begin{corollary} \label{thm:output_entropies}
Let $K\geq 1$. Then for any density matrix $\rho\in B(\irrep_J)$:
\begin{align}
    S_{vN}(\Phi^{K+i}_{J,K}(\rho)) = \log\left(\frac{2K+1}{2J+1}\right)-(2J+1)\int_{S^2}\Hus_J^{-i}(\rho)\log \Hus_J^{-i}(\rho) \domega +\; \mathcal{E},
\end{align}
with $\abs{\mathcal{E}}\leq C\log(2K+1)\tfrac{(2J+1)(J-\abs{i}+1)}{2K-J+i+1}$.

If in addition $K\geq 2J$:
\begin{align}\label{eq:entropy_of_channel_approximation}
    S_{vN}(\Phi^{(\lambda)}_{J,K}(\rho)) = \log\left(\frac{2K+1}{2J+1}\right)-(2J+1)\int_{S^2}\Hus_J^{(\lambda)}(\rho)\log \Hus_J^{(\lambda)}(\rho) \domega +\; \mathcal{E},
\end{align}
with $\abs{\mathcal{E}}\leq C\log(2K+1)\tfrac{(2J+1)^2}{2K+1}$.
\end{corollary}

One of the problems motivating this work is to compute minimal output entropies of $\SU(2)$-equivariant quantum channels. Formula \eqref{eq:entropy_of_channel_approximation} gives the minimal output entropy of channels $\Phi_{J,K}^{(\lambda)}$, up to errors vanishing for $K \to \infty$, in terms of the minimum of the functional
\begin{equation}
    \rho \mapsto - (2J+1) \int_{S^2} \Hus_J^{(\lambda)}(\rho) \log \Hus_J^{(\lambda)}(\rho) \domega,
    \label{eq:generalized_classical_entropy}
\end{equation}
defined on the set of density matrices in $B(\irrep_J)$. The problem of minimizing \eqref{eq:generalized_classical_entropy} in the special case $\Hus_J^{(\lambda)} = \Hus_J$ was raised in \cite{lieb_proof_1978}, where a similar problem for Glauber coherent states was solved. Lieb conjectured in \cite{lieb_proof_1978} that for $\Hus_J^{(\lambda)} = \Hus_J$, the minimizers of \eqref{eq:generalized_classical_entropy} would be spin coherent states. This problem attracted attention of several researchers \cite{schupp_liebs_1999,bodmann_lower_2004}, but the first proof that coherent states are minimizers \cite{lieb_proof_2014}, due to Lieb and Solovej, appeared 36 years after the formulation of the problem (see also \cite{lieb_proof_2016} for a~generalization from $\SU(2)$ to $\SU(N)$). 

The reasoning in \cite{lieb_proof_2014} proceeds by showing that the eigenvalues of any output of $\Phi_{J,K}^{K-J}$ are majorized by the eigenvalues of the image of a coherent state projector. In particular any concave function of the channel output is minimized by coherent states. The relation between the entropy of $\Phi_{J,K}^{K-J} (\rho)$ and the classical entropy of $\Hus_J(\rho)$ then implies Lieb's conjecture, but because of the limit involved in the argument one cannot exclude existence of other minimizers. 

The proof of Lieb's conjecture was finished with the input of Frank \cite{frank_sharp_2023}, who proved that coherent states are the only minimizers. Frank adapted a~method developed by Kulikov \cite{kulikov_functionals_2022} to solve a similar problem for the group $\SU(1,1)$.

A natural question is whether coherent states minimize the output entropy of all $\SU(2)$-equivariant channels. The answer to this question is negative, as we verified by performing brute force minimization for channels with small input dimension (but arbitrary output dimension). Motivated by \eqref{eq:entropy_of_channel_approximation} we were also tempted to guess that rank one projectors $ \rho =\ketbra{ \omega ; i }{ \omega ; i }{J}$ were minimizers of the functional \eqref{eq:generalized_classical_entropy} in the case $\Hus_J^{(\lambda)} = \Hus_J^{i}$, but this is not true either. It seems that the landscape of minimal output entropies, as well as minimizers of classical entropies \eqref{eq:generalized_classical_entropy} is quite rich. It is currently not clear whether the latter minimization problem is actually easier than computation of the minimal output entropy of $\Phi_{J,K}^{(\lambda)}$.

During the final stages of this project there appeared a very nice work of Van Haastrecht \cite{van_haastrecht_limit_2024} containing results partially overlapping with ours. It established strong convergence $\Hus_J \Op_J \to 1$ on $C(S^2)$ and asymptotic trace formulas for channels outputs similar to those in our \Cref{thm:output_traces}. Van Haastrecht \cite{van_haastrecht_limit_2024} showed that the error term $\mathcal E$ in \eqref{eq:trace_formula_channels_with_error_term_intro} converges to $0$ as $J \to \infty$ for any continuous function $\varphi$. In fact, \eqref{eq:trace_formula_channels_with_error_term_intro} is equivalent to the main theorem in \cite{van_haastrecht_limit_2024}. The map $E_{\mu,k}$ that appears in \cite{van_haastrecht_limit_2024} is, in our notation, $(2J+1)^{-1}\Hus^{-i}_J \Op_J$, with $\mu=2J$ and $k=J+i$. Our results, which were obtained independently and use different techniques, improve on those of \cite{van_haastrecht_limit_2024} in some aspects. Firstly, we obtain approximations of channel outputs (\Cref{thm:channel_approx}), not only trace formulas (which amount to computing the spectrum). We also identify the approximate upper symbols of outputs of channels $\Phi_{J,K}^{(\lambda)}$ as the generalized Husimi functions $\Hus_J^{(\lambda)}$, which we find appealing: our results relate directly equivariant quantum channels to equivariant POVMs, with extreme channels corresponding to the extreme POVMs. Lastly, and this is in our view the main advantage of our work, we provide quantitative error bounds in all our main results.

\section{Preparations} \label{sec:prep}

Throughout proofs in this manuscript, $C$ is a positive constant whose value may change from line to line. 

\subsubsection*{Irreducible representations}

We consider the group $\SU(2)$. Its Lie algebra has a standard basis $i S_x,i S_y,i S_z$; we will also use $S_{\pm} = S_x \pm i S_y$, satisfying $[S_z, S_{\pm}] = \pm S_{\pm}$ and $[S_+, S_-] = 2 S_z$. For each non-negative half-integer $J$ we let $\irrep_J$ be the irreducible representation of $\SU(2)$ of dimension $2J+1$. $\irrep_J$ has an orthonormal eigenbasis $(\ket{m}_J)_{m=-J}^J$ such that $\ket{m}_J$ is an eigenvector of $S_z$ with eigenvalue $m$ and we have $S_\pm \ket{m}_J = \sqrt{(J \mp m) (J \pm m +1)} \ket{m \pm 1}_J$. 

An anti-unitary operator $U_J : \irrep_J \to \irrep_J$ commuting with the action of $\SU(2)$ may be defined by
\begin{equation}
U_J \sum_{m=-J}^J\alpha_m\ket{m}_J = \sum_{m=-J}^J(-1)^{J-m} \overline{\alpha}_m\ket{-m}_J.
\end{equation}
It satisfies $U_J^2 = (-1)^{2J}$. An invariant bilinear form $\beta$ on $\SU(2)$ is defined by 
\begin{equation}
    \beta(\ket{\psi}, \ket{\phi}) = \braket{U_J \psi}{ \phi}.
    \label{eq:beta_def}
\end{equation}
It is symmetric if $J$ is an integer and skew-symmetric otherwise. Since $\beta$ is nondegenerate, the dual space of $\irrep_J$, with the standard (contragradient) action of $\SU(2)$ \eqref{eq:contragredient}, is isomorphic to $\irrep_J$ as a representation. The transpose of an operator $\rho\in B(\irrep_J)$ with respect to $\beta$ is defined by 
\begin{align}\label{eq:definition_of_beta}
    \beta(\ket{\phi},\rho \ket{\psi}) = \beta(\rho^\beta \ket{\phi}, \ket{\psi}).
\end{align}
Equivalently, $\rho^\beta = U_J^{\phantom1} \rho^* U_J^{-1}$. 

In any representation of the Lie algebra of $\SU(2)$ we can define the Casimir operator
\begin{equation}
    Q = S_x^2 + S_y^2 + S_z^2 = S_z^2 + \frac{1}{2} S_+ S_- + \frac{1}{2} S_- S_+.
\end{equation}
$Q$ acts on $\irrep_J$ as multiplication by $J(J+1)$. In some calculations we will use the (invertible) operator $1+4Q$, which acts on $\irrep_J$ as multiplication by $(2J+1)^2$.

We are more interested in the action of the Casimir on the reducible representation $B(\irrep_J)$. In order to distinguish the Casimir $B(\irrep_J) \to B(\irrep_J)$ from multiplication by $Q$ regarded as an operator on $\irrep_J$, we introduce the~different font:
\begin{align}
    \cQ \rho &= [S_x,[S_x,\rho]] + [S_y,[S_y,\rho]] + [S_z,[S_z,\rho]] \label{eq:ad_Casimir}  \\
    & = 2J(J+1) \rho - 2 S_z \rho S_z - S_+ \rho S_- - S_- \rho S_+    .\nonumber
\end{align}
Since $S_x,S_y,S_z$ have operator norms $J$, directly from \eqref{eq:ad_Casimir} we have the bound
\begin{equation}
    \| \cQ \rho \|_p \leq  C J^2 \| \rho \|_p,
    \label{eq:Casimir_bound}
\end{equation}
where $\| \cdot \|_p$ is the Schatten $p$-norm on $B(\irrep_J)$, defined as $\| \rho \|_p = \big( \Tr \big[ (\rho^* \rho)^{p/2} \big] \big)^{1/p}$ for $p \in [ 1, \infty)$, and as the operator norm if $p = \infty$. The (non-optimal) constant $C$ in \eqref{eq:Casimir_bound} may be taken to be independent of $p$ and $J$.

For a tensor product of representations we have the Clebsch-Gordan decomposition:
\begin{align}
    \irrep_J\tensor \irrep_K \cong \bigoplus_{M=\abs{K-J}}^{K+J}\irrep_M,
    \label{eq:CG_decomp}
\end{align}
which is multiplicity free. Therefore, for every $M \in \{ |K-J| , \dots, K+J \}$ there exists an equivariant isometric embedding $\iota_{J,K}^M : \irrep_M \to \irrep_J \otimes \irrep_K$, unique up to a~complex phase factor. We let $q_{J,K}^M$ be the adjoint of $\iota_{J,K}^M$ and let $P^M_{J,K} = \iota_{J,K}^M q_{J,K}^M$. That is, $P^M_{J,K}$ is the orthogonal projection onto the $M$th summand in \eqref{eq:CG_decomp}. 

Matrix elements of $\iota^M_{J,K}$, i.e.\ the numerical coefficients in
\begin{equation}
    \iota^M_{J,K} \ket{m}_M = \sum_{j} C_{J,j; K,m-j}^{M,m} \ket{j}_J \otimes \ket{m-j}_K,
    \label{eq:CG_def}
\end{equation}
are called the Clebsch-Gordan coefficients. They can be taken real, and this reduces the freedom of choice of $\iota^M_{J,K}$ down to a sign. The standard choice which fixes $\iota^M_{J,K}$ completely is characterized by the condition $C_{J,J; K, M-J}^{M,M} > 0$. Below we give an upper bound satisfied by Clebsch-Gordan coefficients which will be useful later.

\begin{lemma} \label{lem:CG_bounds}
Let $c_l = C^{M,M}_{J,M-K+l;K,K-l}$, $l \in \{ 0 , \dots ,J+K-M \}$, and let
\begin{equation}
    \epsilon = \frac{(J-M+K)(J+M-K+1)}{-J+M+K+1}.
\end{equation}
Then
\begin{equation}
    |c_l|^2 \leq \epsilon^l |c_0|^2  , \qquad 0 \leq 1 - |c_0|^2 \leq \epsilon.  
\end{equation}
\end{lemma}
\begin{proof}
 Since $\iota^M_{J,K} \ket{M}_M$ is annihilated by $S_+$, we have the recursion relation:
\begin{equation}
c_{l+1} = - \sqrt{\frac{(J-M+K-l)(J+M-K+l+1)}{(l+1)(2K-l)}} c_l,
\end{equation}
and thus:
\begin{equation}
    |c_{l+1}|^2 = \frac{(J-M+K-l)(J+M-K+l+1)}{(l+1)(2K-l)} |c_l|^2,
\end{equation}
for $l \leq J+K-M-1$. We bound:
\begin{equation}
 2K-l \geq -J +M+K+1, \qquad J-M+K-l \geq J-M+K,   
\end{equation}
and:
\begin{equation}
    \frac{J+M-K+l+1}{l+1} \leq J + M -K +1,
\end{equation}
getting $|c_{l+1}|^2 \leq \epsilon |c_l|^2$, and as a consequence:
\begin{equation}
    |c_l|^2 \leq \epsilon^l |c_0|^2.
\end{equation}
Next, assuming $\epsilon<1$,
\begin{equation}
    1 = \sum_l |c_l|^2 \leq |c_0|^2 \sum_l \epsilon^l \leq |c_0|^2 \frac{1}{1 - \epsilon},
\end{equation}
which implies $1-|c_0|^2 \leq \epsilon$. The same bound is also trivially satisfied for $\epsilon \geq 1$. 
\end{proof}

The following special case of \eqref{eq:CG_decomp} plays an important role:
\begin{equation}
    B(\irrep_J) \cong \irrep_J \otimes \irrep_J^* \cong \irrep_J \otimes \irrep_J \cong \bigoplus_{M=0}^{2J} \irrep_M.
\end{equation}

\subsubsection*{Analysis on the sphere}

We define $L^p$ norms\footnote{We use the same notation for $L^p$ norms on $S^2$ and $\SU(2)$ as for the Schatten norms of operators. We hope it is always clear from the context which is meant.} $\| \cdot \|_p$ of functions on the sphere $S^2$ in terms of the Euclidean measure with total area normalized to $1$. For $p=2$ we have the orthogonal decomposition
\begin{equation}
    L^2(S^2) \cong \bigoplus_{L=0}^\infty \mathcal H_L,
    \label{eq:sphere_decomp}
\end{equation}
in which only integer $L$ occur. The vector fields on $S^2$ representing the Lie algebra of $\SU(2)$ will be denoted $L_i$. The Casimir on $L^2(S^2)$ is $-\Delta$, where $\Delta$ is the Laplacian induced by the standard metric on the unit sphere. In terms of the $L_i$,
\begin{equation}
    - \Delta = L_x^2 + L_y^2 + L_z^2. 
\end{equation}
The projection $\Pi_L$ on the $L$th summand in \eqref{eq:sphere_decomp} is explicitly given by
\begin{equation}
    \Pi_L f(\omega) = \int (2L+1) P_L(\omega \cdot \omega') f(\omega') \domega ',
\end{equation}
where $P_L$ is the $L$th Legendre polynomial. We remark that this formula differs by a~factor $4 \pi$ from what can be found in some references because $\frac{1}{4 \pi}$ was absorbed in the definition of $\domega$.

In terms of the standard spherical cooordinates $(\theta, \phi)$ on $S^2$, the Poisson bracket of functions $f$ and $g$ takes the form:
\begin{equation}
    \{ f, g \} = \frac{1}{\sin \theta} \left( \frac{\partial f}{\partial \theta} \frac{\partial g}{\partial \phi} - \frac{\partial f}{\partial \phi} \frac{\partial g}{\partial \theta}\right). \label{eq:poisson_bracket} 
\end{equation}

The Sobolev space $W^{k,p}(S^2)$ with $k \in \mathbb N$ and $p \in [1,\infty]$ is defined as the space of functions in $L^p(S^2)$ whose (distributional) covariant derivatives up to order $k$ are in $L^p(S^2)$. For the norm on $W^{k,p}(S^2)$ we take:
\begin{align}
    \| f \|_{k,p}  = \begin{cases}\left( \sum_{n=0}^k \| \nabla^n f \|_p^p \right)^{\frac{1}{p}}, \qquad & p \neq \infty, \label{eq:Sobolev_norm} \\
    \max \{ \| \nabla^n f \|_\infty  \}_{n=0}^k , \qquad & p = \infty  .\end{cases}  
\end{align}
We remark that functions in $W^{2,1}(S^2)$ are continuous \cite{adams_cone_1977}. We consider also the space $C^{0,\alpha}(S^2)$ of $\alpha$-H\"older continuous functions, with the seminorm 
\begin{equation}
\norm{f}_{C_h^{0,\alpha}} = \sup_{\omega,\omega' \in S^2} \frac{|f(\omega)-f(\omega')|}{| \omega - \omega' |^{\alpha}}, \label{eq:Holder_seminorm}
\end{equation}
where $|\cdot|$ is the length of a vector in $\mathbb R^3$. 

\subsubsection*{Analysis on the group}

We have two commuting actions of $\SU(2)$ on function spaces $L^p(\SU(2))$: the left regular representation,
\begin{equation}
    (\pi_L(g) \psi)(h) = \psi(g^{-1} h ),
\end{equation}
and the right regular representation,
\begin{equation}
    (\pi_R(g) \psi)(h) = \psi (hg).
\end{equation}
Differentiating at the identity one obtains two sets of vector fields (first order differential operators),
\begin{equation}
    \xi_i^R =  \left. -i \frac{\dd}{\dd t} \pi_L(e^{ i t S_i}) \right|_{t=0}, \qquad -\xi_i^L =  \left. -i \frac{\dd}{\dd t} \pi_R(e^{i t S_i}) \right|_{t=0}, \qquad i = x , y , z, 
\end{equation}
spanning two copies of the Lie algebra of $\SU(2)$. Vector fields $\xi_i^L$ and $\xi_i^R$ have equal values at the neutral element of the group. Since the left and right actions commute, $\xi_i^R$ are right-invariant: $\xi_i^R \pi_R(g) = \pi_R(g) \xi_i^R$ for every $g \in \SU(2)$. Similarly, $\xi_i^L$ are left-invariant. The Laplacian on $\SU(2)$ is the Casimir for either of the two representations:
\begin{equation}
    - \Delta = (\xi_x^L)^2+(\xi_y^L)^2+(\xi_z^L)^2 = (\xi_x^R)^2+(\xi_y^R)^2+(\xi_z^R)^2.
\end{equation}
$\Delta$ is also the Laplacian of an invariant Riemannian metric on $\SU(2)$.

The sphere $S^2$ may be identified with the quotient space $\SU(2) / \mathrm{U}(1)$. That is, we identify the left coset $h \mathrm{U(1)}$ with the point $h \cdot \uparrow$, where $\uparrow$ is the north pole of $S^2$. We let $q : \SU(2) \to S^2$ be the quotient map. Right invariant fields descend to the quotient, in the sense that
\begin{equation}
    \xi_i^R (f \circ q) = (L_i f) \circ q
\end{equation}
for a function $f$ on $S^2$. 

Let $X$ be a Banach space on which $\SU(2)$ acts (strongly continuously) preserving the norm. We denote the application of $g \in \SU(2)$ to $ x \in X$ by $gx$. If $\varphi \in L^1(\SU(2))$, let:
\begin{equation}
    \varphi \rhd x = \int_{\SU(2)} \varphi(g) gx \dg, 
    \label{eq:conv_rep}
\end{equation}
where the integration is done with respect to the Haar probability measure on $\SU(2)$. The integral converges and we have:
\begin{equation}
    \| \varphi \rhd  x \| \leq \| \varphi \|_1 \| x \| . \label{eq:norm_bound_rhd}
\end{equation}

The action \eqref{eq:conv_rep} of $\varphi$ on $\psi$ in the left regular representation is the convolution:
\begin{equation}
    (\varphi * \psi)(h) = \int_{\SU(2)} \varphi(g) \psi(g^{-1} h) \dg  .
\end{equation}
If $\varphi, \psi \in L^1(\SU(2))$, then
\begin{equation}
    \varphi \rhd (\psi \rhd x) = (\varphi * \psi) \rhd x.
\end{equation}

A special role is played by class functions, i.e.\ functions $\varphi$ on $\SU(2)$ such that $\varphi(ghg^{-1}) = \varphi(h)$ for all $g,h \in \SU(2)$. If $\varphi$ is an integrable class function, then operators $\varphi \rhd$ commute with the action of $\SU(2)$, and $\varphi * \psi = \psi * \varphi$ holds for any $\psi \in L^1(\SU(2))$. An orthonormal eigenbasis of the space of square-integrable class functions is given by the characters of irreducible representations:
\begin{equation}
    \chi_J(g) = \Tr_{\irrep_J} (g).
\end{equation}
On a representation with orthogonal decomposition $\bigoplus_{M } \irrep_M^{\oplus \mu_M}$, the projection $\Pi_M$ on the $M$th summand (often called an isotypic component) is given by
\begin{equation}
    \Pi_M x = (2M+1) \chi_M \rhd x. 
\label{eq:projection_character}
\end{equation}

Every $g \in \SU(2)$ has eigenvalues $e^{ \pm i \alpha}$. To fix $\alpha$ uniquely we require $\alpha \in [0, \pi]$. Two elements $g,h \in \SU(2)$ are conjugate (i.e. $g = khk^{-1}$ for some $k \in \SU(2)$) if and only if they have the same eigenvalues. This justifies the following definition: for every class function $\varphi$ we let $R \varphi$ be the unique function on $[0, \pi]$ such that $\varphi(g) = R \varphi(\alpha)$ for $g$ with eigenvalues $e^{\pm i \alpha}$. For example, $R \chi_J(\alpha )= \frac{\sin((2J+1) \alpha)}{\sin (\alpha)}$. 

Integrals of class functions can be computed using Weyl's integration formula:
\begin{equation}
    \int_{\SU(2)} \varphi(g) \dg = \frac{2}{\pi} \int_0^\pi R \varphi(\alpha) \sin^2 (\alpha) \dd \alpha.
    \label{eq:Weyl_integration}
\end{equation}
Moreover, if $\varphi$ is a class function, then $\Delta \varphi$ is also a class function. We have the formula
\begin{equation}
    (R (1 - 4 \Delta) \varphi)(\alpha) = - \frac{1}{\sin (\alpha)} \frac{\dd^2}{\dd \alpha^2} \sin(\alpha) R \varphi(\alpha).
\label{eq:class_function_Laplacian}
\end{equation}
Using this formula we find a class function $G \in L^1(G)$ such that $(1-4 \Delta) G $ is the Dirac measure at the neutral element:
\begin{equation}
    RG (\alpha) = \frac{1}{2} \frac{\pi - \alpha}{\sin (\alpha)}. 
\end{equation}
We have $G \rhd = (1+4Q)^{-1}$. For example,
\begin{equation}
    (1 + 4 \cQ)^{-1} (\rho) = \int_{\SU(2)} G(g) g \rho g^{-1} \dg, \qquad \text{for } \rho \in B(\irrep_J).
\end{equation}
We will also use the symmetrization of $G$ with respect to the transformation $g \mapsto - g$:
\begin{align}
    G^{\mathrm{SO}}(g) &= \frac{1}{2} (G(g) + G(-g)) , \\
    RG^{\mathrm{SO}} (\alpha) & = \frac{1}{2} \left( RG(\alpha) + RG(\pi - \alpha) \right) 
 = \frac{\pi}{4} \frac{1}{\sin (\alpha)}. \nonumber 
\end{align}
We have $G \rhd = G^{\mathrm{SO}} \rhd$ on representations of $\mathrm{SO}(3)$, but $G^{\mathrm{SO}} \rhd$ annihilates $\irrep_J$ with $J \not \in \Z$. In particular,
\begin{equation}
    (1 + 4 \cQ)^{-1} (\rho) = \int_{\SU(2)} G^{\mathrm{SO}}(g) g \rho g^{-1} \dg, \qquad \text{for } \rho \in B(\irrep_J).
\end{equation}

\section{Quantum to classical maps and quantization} \label{sec:quantization}

Let $M$ be a measurable space. A positive operator-valued measure (POVM) on $M$ is a map $\Pi$ from measurable subsets of $M$ to bounded operators on some Hilbert space $\mathcal H$ such that $\bra{\psi} \Pi(\cdot) \ket{\psi}$ is a probability measure for every unit vector $\ket{\psi} \in \mathcal H$.

If $\Psi$ is a linear map from $B(\mathcal H)$ to $\mathcal M(M)$, the space of complex measures on $M$, which takes density matrices to probability measures, we call $\Psi$ a quantum to classical map. A POVM $\Pi$ induces a quantum to classical map $\rho \mapsto \Tr(\rho \Pi(\cdot))$, and every quantum to classical map arises this way from a unique POVM. The adjoint of $\Psi$ is the map $f \mapsto \int_M f(m) \dd \Pi(m)$, taking measurable functions to operators.

We are interested in equivariant quantum to classical maps with $M$ being a classical phase space (i.e.\ a symplectic manifold) with a~transitive action of $\SU(2)$. This singles out the homogeneous space $\SU(2)/\mathrm{U}(1)\cong S^2$, which is also the orbit of the projection onto the highest weight vector $\ket{J}_J\in \irrep_J$. The prototypical example of a quantum to classical map for the sphere is the Husimi map $\Hus_J$. We recall the definition of generalized Husimi maps $\Hus_J^i : B(\irrep_J) \to C^\infty(S^2)$:
\begin{equation}
\Hus_J^i (\rho) (\omega) =  {}_J \bra{\omega ; i } \rho \ket{\omega ; i}_J.
\end{equation}

Equivariant POVMs on homogeneous spaces of Lie groups were studied by many authors, e.g. \cite{davies_repeated_1970,holevo_probabilistic_2011,chiribella_extremal_2004}. The situation for $\SU(2)$ is relatively simple, as reviewed below. 

\begin{proposition} \label{prop:qtc_simplex}
$\SU(2)$-equivariant quantum to classical maps $B(\irrep_J) \to \mathcal M(S^2)$ form a~simplex with vertices $(2J+1) \Hus^i_J$, $i \in \{ - J, \dots, J \}$. 
\end{proposition}

\begin{proof}
Let $\Psi\colon B(\irrep_J)\to \mathcal{M}(S^2)$ be a quantum to classical map. $\Psi(1)$ is an invariant measure of total mass $2J+1$, so it equals $(2J+1) \domega$. For a positive operator $\rho$ we have 
\begin{equation}
 \Psi(\rho) \leq \| \rho \| (2J+1) \domega  , 
\end{equation}
so~$\Psi(\rho) = f_\rho d \omega$ for some $f_\rho \in L^\infty(S^2, \domega)$ with $\| f_\rho \|_{\infty} \leq \| \rho \| (2J+1)$. By linearity of~$\Psi$, $\Psi(\rho) = f_\rho \domega $ with $f(\rho) \in L^\infty(S^2, \domega)\subseteq L^2(S^2,\domega)$ holds for every $\rho$. The map $\rho\mapsto f_\rho$ is equivariant, so by Schur's lemma its image is contained in the subspace $\bigoplus_{\ell = 0}^{2J}\irrep_\ell\subset L^2(S^2)$. In particular, $f_\rho$ is smooth. 

The map $\rho \mapsto f_\rho(\omega)$ is a positive linear functional taking $1$ to $2J+1$, so it is given by $\Tr(A(\omega) \rho)$ for some positive operator $A(\omega)$ with $\Tr(A(\omega))= 2J+1$. By equivariance, $A(\cdot)$ is uniquely determined by its value at the north pole $\omega = \,\uparrow$. Let $A = A(\uparrow)$. Since $\uparrow$ is invariant under rotations around the $z$ axis, $A$ commutes with $S_z$. Hence it is a~convex combination
\begin{equation}
    A = (2J+1) \sum_i \lambda_i \ket{i} \bra{i}, \qquad \lambda_i \geq 0 , \ \sum_i \lambda_i=1,
\end{equation}
from which we can conclude
\begin{equation}
    \Psi = (2J+1) \sum_i \lambda_i \Hus_J^i.
\end{equation}
\end{proof}

The adjoint of $(2J+1) \Hus_J^i$ is the map $\Op_J^i$, 
\begin{equation}
    \Op^i_J (f) = (2J+1)\int_{S^2} f(\omega) \ketbra{\omega;i}{\omega;i}{J} \domega.
\end{equation}
$\Op^i_J (f)$ is defined even for a distribution $f$, but to make useful statements we will need some regularity of $f$.

The following properties are immediate:

\begin{lemma}
$\Op^i_J$ is an equivariant map with the properties
\begin{enumerate}[label=\roman*)]
    \item $\Op^i_J(f)^*=\Op^i_J(\overline{f})$;
    \item $\Op^i_J$ is positivity preserving, that is $\Op^i_J(f)\geq 0$ whenever $f\geq 0$ a.e.;
    \item if $f$ is real valued, then $\sigma(\Op^i_J(f))\subseteq \image f$;
    \item if $f \in L^1(S^2)$, then 
    \begin{align}
        \Tr(\Op^i_J(f)) = (2J+1)\int_{S^2} f(\omega) \domega;
        \label{eq:Op_trace}
    \end{align}
    \item $\Op^i_J(1) = \identity$.
\end{enumerate}
\end{lemma}

Next we provide estimates on Schatten norms of $\Op^i_J(f)$. We remark that essentially the same proof appears in \cite{simon_classical_1980}.

\begin{lemma} \label{lem:oph_norms}
For $J\in \tfrac{1}{2}\N$ and any $p\in [1,\infty]$:
\begin{align}
    \norm{\Op^i_J(f)}_p \leq (2J+1)^{1/p}\norm{f}_p,\quad f\in L^p(S^2),
\label{eq:Op_bound}
\end{align}
and the inequalities are saturated for the constant function $f\equiv 1$. Similarly,
\begin{equation}
    \| \Hus_J^i (\rho) \|_p \leq (2J+1)^{- 1/p} \| \rho \|_p,
    \label{eq:Hus_bound}
\end{equation}
with equality if $\rho=1$.
\end{lemma}

\begin{proof}
Bounds on $\Hus_J^i$ are obtained from these on $\Op_J^i$ by taking adjoints. We consider $\Op_J^i$. The $1\to 1$ estimate follows from subadditivity of norms:
\begin{align}
    \norm{\Op^i_J(f)}_1\leq (2J+1)\int_{S^2} \abs{f(\omega)}\norm{\ketbra{\omega;i}{\omega;i}{}}_1 \domega = (2J+1) \int_{S^2}\abs{f(\omega)}\domega.
\end{align}
For the $\infty \to \infty$ estimate, notice that, for any $\ket{\psi}, \ket{\phi}\in \irrep_J$:
\begin{align}
\abs{\bra{\phi}\Op^i_J(f)\ket{\psi}}&\leq (2J+1)\int_{S^2}\abs{ \braket{\phi}{\omega;i} } \abs{ \braket{\omega;i}{\psi} } 
 \abs{f(\omega)}\domega\\
&\leq \norm{f}_\infty (2J+1) \left(\int_{S^2} \abs{\braket{\phi}{\omega;i}}^2\domega\right)^{\frac{1}{2}}\left(\int_{S^2}\abs{\braket{\omega;i}{\psi}}^2\domega\right)^{\frac{1}{2}} \nonumber \\
&=\norm{f}_\infty \norm{\phi}\norm{\psi}. \nonumber 
\end{align}
Thus $\norm{\Op^i_J(f)}_\infty\leq \norm{f}_\infty$.
The result for $p\in (1,\infty)$ follows from interpolation.
\end{proof}

\subsection{Inversion} \label{sec:inversion}

In this Subsection we consider operators $\Hus_J \Op_J$ and $\Op_J \Hus_J$. These operators are positive, positivity-preserving, and have $p \to p$ norm equal $1$ for every $p \in [1 , \infty]$. We will describe their convergence to identity as $J \to \infty$.

We remark that the formula \eqref{eq:Hus_Op_coeff} below appears in slightly different forms in \cite{varilly_moyal_1989,zhang_berezin_1998}, and that similar formulas for the compositions $\Hus_J^i \Op_J^i$ were studied in \cite{shmoish_spectrum_2021}. 

\begin{lemma}\label{lemma:HJ_OpJ_composition}
Let $J\in \tfrac{1}{2}\N$. Then:
\begin{align}
    \Hus_J\Op_J &= \sum_{\ell=0}^{2J} \frac{(2J)!(2J+1)!}{(2J-\ell)!(2J+\ell+1)!}\Pi_\ell, \label{eq:Hus_Op_coeff} \\
    \Op_J  \Hus_J &= \sum_{\ell=0}^{2J} \frac{(2J)!(2J+1)!}{(2J-\ell)!(2J+\ell+1)!}\Pi_\ell,
\end{align}
where $\Pi_\ell$ is the orthogonal projection onto $\irrep_\ell$ (in~$L^2(S^2)$ or in $B(\irrep_J)$). 
\end{lemma}

\begin{proof}
It will be sufficient to show the formula for the composition $\Hus_J \Op_J$, since $\Hus_J$ and $\Op_J$ are adjoint to each other up to a multiplicative constant.

We write $\Hus_J \Op_J$ as:
\begin{align}
    \Hus_J\Op_J (f)(\omega) &  = (2J+1) \int_{S^2}f(\omega')\abs{\braket{\omega}{\omega'}}^2\domega'  \\
    & = (2J+1)\int_{S^2}f(\omega')\Big(\frac{1+\cos(\theta)}{2}\Big)^{2J}\domega', \nonumber
\end{align}
where $\cos(\theta)=\omega\cdot \omega'$. We expand the cosine factor in terms of Legendre polynomials:
\begin{align*}
\Big(\frac{1+\cos(\theta)}{2}\Big)^{2J} = \sum_{\ell=0}^{2J} c_{J,\ell} (2\ell+1)P_\ell (\cos(\theta)), 
\end{align*}
where $2c_{J,\ell} = \int_{-1}^1 ((1+x)/2)^{2J}P_\ell(x)\dx$. This integral can be explicitly calculated using the Rodriguez formula:
\begin{align*}
    P_\ell(x) = \frac{1}{2^\ell \ell!} \frac{d^\ell}{dx^\ell} (x^2-1)^\ell,
\end{align*}
from which, integrating by parts $l$ times, we obtain that:
\begin{align*}
\int_{-1}^1 \Big(\frac{1+x}{2}\Big)^{2J}P_\ell(x)\dx = \frac{2(2J)!^2}{(2J-\ell)!(2J+\ell+1)!}.
\end{align*}
 The desired formula follows from this, since:
 \begin{align*}
(\Hus_J\Op_J f)(\omega) &= (2J+1) \int_{S^2} f(\omega') \sum_{\ell=0}^{2J}\frac{(2J)!^2}{(2J-\ell)!(2J+\ell+1)!}P_\ell(\omega \cdot \omega')\domega'\\
     &= \Big[\sum_{\ell=0}^{2J}\frac{(2J)!(2J+1)!}{(2J-\ell)!(2J+\ell+1)!}\Pi_\ell \Big] f (\omega),
 \end{align*}
as desired.
\end{proof}

\begin{proposition}\label{prop:HJ_OpJ_approximate_inversion_2_norm}
The sequence of operators $\iset{\Hus_J \Op_J}_{J}$ is monotonically increasing in $B(L^2(S^2))$. 
For~every $s \in [0,1]$ we~have the bounds
\begin{align}
    \| (1-\Hus_J \Op_J ) f \|_2 & \leq \frac{1}{(2J+1)^s} \| (-\Delta)^s f \|_2,  \\
    \| (1- \Op_J \Hus_J) \rho \|_2 & \leq \frac{1}{(2J+1)^s} \| \cQ^s \rho \|_2,
\end{align}
and
\begin{align}
        \norm{\left(1-\Hus_J \Op_J+\frac{\Delta}{2J+1}\right)f}_{2} & \leq\frac{1}{(2J+1)^{1+s}} \| 
(-\Delta)^{1+s} f\|,  \\
\norm{\left(1- \Op_J \Hus_J-\frac{\cQ}{2J+1}\right)\rho}_{2} 
& \leq\frac{1}{(2J+1)^{1+s}} \| 
\cQ^{1+s} \rho \|.
    \end{align}
\end{proposition}
\begin{proof}
That the sequence is monotonically increasing and coverges to $\identity$ follows immediately from \Cref{lemma:HJ_OpJ_composition} and the fact that, for fixed $\ell$, the sequence
\begin{align}
    \frac{(2J)!(2J+1)!}{(2J-\ell)!(2J+\ell+1)!},\quad J\in \frac{1}{2}\N,
\end{align}
is monotonically increasing with limit $1$. We have a bound
\begin{equation}
    0\leq 1-\frac{(2J)!(2J+1)!}{(2J-\ell)!(2J+\ell+1)!} \leq \frac{\ell(\ell+1)}{2 J+\ell+1},
\end{equation}
and hence $1-\frac{(2J)!(2J+1)!}{(2J-\ell)!(2J+\ell+1)!} \leq \left( \frac{\ell (\ell+1)}{2J+1} \right)^s$ for every $s \in [0,1]$. Thus
\begin{equation}
\|(\Hus_J \Op_J -1) f\|_2^2 \leq \sum_{\ell=0}^{\infty} \left( \frac{\ell (\ell+1)}{2J+1} \right)^{2s} \| \Pi_\ell f \|_2^2  \leq \frac{1}{(2J+1)^{2s}} \| (-\Delta)^s f \|_2^2,
\end{equation}
and similarly for $\Op_J \Hus_J$. The second bound will follow if we prove:
    \begin{align}
        \abs{1-\frac{(2J)!(2J+1)!}{(2J-\ell)!(2J+\ell+1)!}-\frac{\ell(\ell+1)}{2J+1}}\leq \frac{\ell^2(\ell+1)^2}{(2J+1)^2}.
    \end{align}
    Notice that we may assume that $\ell\geq 1$, since the bound is trivial for $\ell=0$. To see that the above holds, we rewrite the left hand side as:
    \begin{align}
        \binom{2J+\ell+1}{\ell+1}^{-1}\abs{\binom{2J+\ell+1}{\ell+1}-\binom{2J+1}{\ell+1}-\ell\binom{2J+\ell+1}{\ell}}.
    \end{align}
    We make use of the identities:
    \begin{align}
        \binom{2J+\ell+1}{\ell+1}-\binom{2J+1}{\ell+1} = \sum_{n = 1}^{\ell} \binom{2J+n}{\ell}
    \end{align}
    and
    \begin{align}
        \ell \binom{2J+\ell+1}{\ell} &= \left[\binom{2J+\ell}{\ell}+\binom{2J+\ell}{\ell-1}\right]\\
        &\quad+\left[\binom{2J+\ell-1}{\ell}+\binom{2J+\ell-1}{\ell-1}+\binom{2J+\ell}{\ell-1}\right]+\ldots\nonumber\\ 
        &\quad+\left[\binom{2J+1}{\ell}+\binom{2J+1}{\ell-1}+\binom{2J+2}{\ell-1}+\ldots+\binom{2J+\ell}{\ell-1}\right]\nonumber\\
        &=\sum_{n = 1}^{\ell} \left[ \binom{2J+n}{\ell}+n\binom{2J+n}{\ell-1} \right] \nonumber
    \end{align}
    where the first identity is a consequence of the hockey-stick identity, and the second is obtained by iterating the recursive definition of binomial coefficients to each of the $\ell$ terms individually.

    Combining the two, and bounding each term in the resulting sum by the largest one:
    \begin{align}
        \binom{2J+\ell+1}{\ell+1}^{-1}\abs{\sum_{n = 1}^{\ell}n\binom{2J+n}{\ell-1}}&\leq \binom{2J+\ell+1}{\ell+1}^{-1}\ell^2\binom{2J+\ell}{\ell-1}\\
        &\leq \frac{\ell^3(\ell+1)}{(2J+\ell+1)(2J+1)}, \nonumber
    \end{align}
    from which the claim follows.
\end{proof}

\begin{lemma}\label{lemma:HJ_OpJ_as_integral_operator}
We have the formulas
\begin{align}
   \Hus_J \Op_J f (\omega) &= \int_{\SU(2)} F_J(g) f(g^{-1} \omega) \dg,    \\ 
   \Op_J \Hus_J \rho &= \int_{\SU(2)} F_J(g) g \rho g^{-1} \dg,
\end{align}
where $F_J \in L^1(\SU(2))$ is given by
\begin{equation}
    F_J(g) = \frac{(2J)!(2J+1)! 2^{4J}}{(4J)!} \cos^{4J} (\alpha)
\end{equation}
for $g \in \SU(2)$ with eigenvalues $e^{\pm i \alpha}$, $\alpha \in [0, \pi]$.
\end{lemma}
\begin{proof}
By
\Cref{lemma:HJ_OpJ_composition} and \eqref{eq:projection_character}, $\Hus_J \Op_J$ and $\Op_J \Hus_J$ are both given as $F_J \rhd$, where \begin{align}
    F_J = \sum_{\ell =0}^{2J} \frac{(2J)!(2J+1)!}{(2J-\ell)!(2J+1+\ell)!} (2\ell+1) \chi_\ell. 
\end{align}
In particular $F_J$ is a class function. We calculate
\begin{align}
    RF_J (\alpha) &= \sum_{\ell =0}^{2J} \frac{(2J)!(2J+1)!}{(2J-\ell)!(2J+1+\ell)!} (2\ell+1) \sum_{m=-\ell}^\ell e^{2 i m  \alpha} \\
    & = \sum_{m=-2J}^{2J}  e^{2 i m  \alpha} \sum_{\ell = |m|}^{2J} \frac{(2J)!(2J+1)!}{(2J-\ell)!(2J+1+\ell)!}  \nonumber \\
    & = \sum_{m=-2J}^{2J}  e^{2 i m  \alpha} \frac{(2J)!(2J+1)!}{(2J-m)! (2J+m)!} = \frac{(2J)!(2J+1)! 2^{4J}}{(4J)!} \cos^{4J} (\alpha). \nonumber
\end{align}
\end{proof}

\begin{proposition}\label{prop:HJ_OpJ_approximate_inversion_all_p}
There exists a constant $C>0$ such that for every $p \in [1 , \infty]$ we have the bounds:
\begin{align}
    \| (1-\Hus_J \Op_J ) f \|_p & \leq \frac{C}{2J+1} \| (1-4\Delta) f \|_p,  \\
    \| (1- \Op_J \Hus_J) \rho \|_p & \leq \frac{C}{2J+1} \| (1+4\cQ) \rho \|_p.
    \end{align}
\end{proposition}
\begin{proof}
The claim is equivalent to $p \to p$ bounds on operators $(1-\Hus_J \Op_J)(1- 4 \Delta)^{-1}$, $(1- \Op_J \Hus_J)(1+ 4 \cQ)^{-1}$. Both these operators are given as $(F_J * G - G^{\mathrm{SO}}) \rhd$, so by \eqref{eq:norm_bound_rhd} it is enough to bound the $L^1$ norm of $F_J * G - G^{\mathrm{SO}}$. 

We find $ F_J*G$ by solving $(1-4\Delta) ( F_J*G)=F_J$ using \eqref{eq:class_function_Laplacian}. We get
\begin{align}
    R(F_J*G)(\alpha) & = \frac{2^{8J}(2J)!^4(2J+1)}{(4J+1)!^2} \sum_{l=0}^{2J} \frac{(\frac{1}{2})_l}{l!} \cos^{2l} \alpha \\
    &= \frac{2^{8J}(2J)!^4(2J+1)}{(4J+1)!^2} \left[ \frac{1}{\sin(\alpha)} -\sum_{l=2J+1}^{\infty} \frac{(\frac{1}{2})_l}{l!} \cos^{2l} \alpha \right]. \nonumber
\end{align}
Next we use \eqref{eq:Weyl_integration} and triangle inequality to bound the $L^1$ norm of $G*F_J - G^{\mathrm{SO}} $:
\begin{align}
    \| F_J*G - G^{\mathrm{SO}} \|_1 \leq & \left| \frac{2^{8J}(2J)!^4(2J+1)}{(4J+1)!^2} - \frac{\pi}{4} \right| \frac{2}{\pi} \int_0^\pi \frac{1}{\sin(\alpha)} \sin^2 (\alpha) \dd \alpha \\
    & + \frac{2^{8J}(2J)!^4(2J+1)}{(4J+1)!^2} \sum_{\ell=2J+1}^\infty \frac{(\frac{1}{2})_l}{l!} \frac{2}{\pi} \int_0^\pi \cos^{2l}(\alpha) \sin^2(\alpha) \dd \alpha. \nonumber
    \end{align}
Computing the integrals we find
\begin{align}
  \| F_J*G - G^{\mathrm{SO}} \|_1 \leq  & \left| \frac{2^{8J}(2J)!^4(2J+1)}{(4J+1)!^2} - \frac{\pi}{4} \right| \frac{4}{\pi} \\
  &+ \frac{2^{8J}(2J)!^4(2J+1)}{(4J+1)!^2} \sum_{\ell=2J+1}^\infty \frac{(\frac{1}{2})_l^2}{l! (l+1)!}. \nonumber
\end{align}
The first term is $O \left( \frac{1}{J} \right)$ by Stirling's formula. We bound the second term:
\begin{align}
    & \sum_{\ell=2J+1}^\infty \frac{(\frac{1}{2})_l^2}{l! (l+1)!} = \frac{1}{\pi} \sum_{\ell=2J+1}^\infty \frac{\Gamma(l + \frac{1}{2})^2}{l! (l+1)!} \leq \frac{1}{\pi} \sum_{\ell=2J+1}^\infty \frac{\Gamma(l) \Gamma(l+1)}{l! (l+1)!} \\
 = & \frac{1}{\pi} \sum_{\ell = 2J+1}^\infty \frac{1}{\ell (\ell+1)} \leq \frac{1}{\pi} \int_{2J}^\infty \frac{1}{\ell (\ell+1)} \dd \ell  = \frac{1}{\pi} \log \left( 1 + \frac{1}{2J} \right)
\end{align}
We conclude that, for some constant $C>0$,
\begin{equation}
    \| F_J*G - G^{\mathrm{SO}} \|_1 \leq \frac{C}{2J+1}.
    \end{equation}
\end{proof}

\subsection{Multiplicativity and traces} \label{sec:mult}

We explain in this section in which sense $\Op_J(fg)\approx \Op_J(f)\Op_J(g)$ for sufficiently regular functions $f, g$. We bound the error of this approximation in Schatten $p$-norms.

$\Op_J(f)$ has eigenvalues of order $1$ if $f$ is of order $1$. Therefore, its Schatten $p$-norm is order $(2J+1)^{\frac{1}{p}}$. We will show that the $p$-norm of $\Op_J(fg)- \Op_J(f)\Op_J(g)$ is bounded by one less power of $J$, i.e. by $(2J+1)^{-1 + \frac{1}{p}}$, for all $f,g$ in a suitable Sobolev space. We~obtain also the next term in the expansion of $\Op_J(f)\Op_J(g)$ for large $J$.

One of our goals is to understand the functional calculus for $\Op_J(f)$, where $f$ is a~real-valued function on the sphere. We show that for an analytic function $\varphi$ and $f$ in a suitable Sobolev space, the operator $\varphi(\Op_J(f))$ can be approximated by $\Op_J(\varphi \circ f)$, with an explicit error bound. Convergence without an explicit rate is obtained for all continuous $\varphi$ and $f$.

In proofs in this Subsection we will use off-diagonal Husimi functions:
\begin{align}
    \Hus_J^{a,b}(A)(\omega) = {}_J\bra{\omega;a}A\ket{\omega;b}_J.
    \label{eq:off-Hus}
\end{align}
There is some ambiguity in this notation: the eigenvectors $\ket{\omega;a}_J$ and $\ket{\omega;b}_J$ of $\omega\cdot S$ are defined at each point $\omega$ only up to a phase factor, and these factors do not cancel in \eqref{eq:off-Hus}. However, the absolute value of $\Hus_J^{a,b}(A)$ as well as the products $\Hus_J^{a,b}(A)\Hus_J^{b,a}(B)$ are well-defined. Using \eqref{eq:Hus_bound} and the Cauchy-Schwarz inequality one can show that
\begin{equation}
    \| \Hus_J^{a,b}(A) \|_p \leq (2J+1)^{- \frac{1}{p}} \| A \|_p.
    \label{eq:off_Hus_bound}
\end{equation}

\begin{lemma}\label{lemma:husimi_of_OpJfOpJg_bound}
    Fix $m\geq 2$ and let $J\in\tfrac{1}{2}\N$ be such that $2J\geq m$. Then:
    \begin{align}
        \Hus_J(\Op_J(f)\Op_J(g)) = \sum_{i=0}^{m-1}\Hus_J^{J,J-i}(\Op_J(f))\Hus_J^{J-i,J}(\Op_J(g)) +\mathcal{E},
    \end{align}
where for each $1\leq p, p_1, p_2\leq \infty$ with $1/p=1/p_1+1/p_2$:
\begin{align}
    \norm{\mathcal{E}}_{p} \leq \frac{C_{m}}{(2J+1)^m}\norm{f}_{m,p_1}\norm{g}_{m,p_2}.
\end{align}
In particular for each $1/p=1/p_1+1/p_2$:
\begin{align}\label{eq:Husimi_product_expansion_first_order}
\Hus_J(\Op_J(f)\Op_J(g)) = fg + \mathcal{E}_1(f,g),
\end{align}
where $\norm{\mathcal{E}_1(f,g)}_p\leq \frac{C}{2J+1}\norm{f}_{2,p_1}\norm{g}_{2,p_2}$, as well as:
\begin{align}\label{eq:Husimi_product_expansion_second_order}
\Hus_J(\Op_J(f)\Op_J(g)) = fg + \frac{1}{2J+1}[\Delta(fg)-\nabla f\cdot \nabla g+i\poisson{f}{g}]+\mathcal{E}_2(f,g),
\end{align}
with $\norm{\mathcal{E}_2(f,g)}_1\leq \frac{C}{(2J+1)^2}\norm{f}_{4,2}\norm{g}_{4,2}$.
\end{lemma}

\begin{proof}
For given $f, g$ we can expand using the resolution of identity in terms of an eigenbasis of $\omega \cdot S$:
\begin{align}\label{eqtn:HJ_multiplicativity_expansion}
    \Hus_J(\Op_J(f)\Op_J(g)) = \sum_{i=0}^{2J} \Hus_J^{J,J-i}(\Op_J(f))\Hus_J^{J-i,J}(\Op_J(g)).
\end{align} 
The left hand side of \eqref{eqtn:HJ_multiplicativity_expansion} can be lifted to a function on $SU(2)$ as $\tilde \Hus_J(A) = \Hus_J(A)\circ q$.
We define $\tilde\Hus_J^{a,b}$ as:
\begin{align}\label{eq:off_diagonal_husimi_group}
\tilde\Hus_J^{a,b}(A)(h) = {}_J\bra{a}h^{-1}Ah\ket{b}_J,\quad h\in \SU(2).
\end{align}
Thus $\tilde\Hus_J^{a,b}$ is a lift of $\Hus_J^{a,b}$, but without a phase ambiguity. In particular,
\begin{align}
        \tilde \Hus_J^{J,J-i}(A)\tilde \Hus_J^{J-i,J}(B) &= \Hus_J^{J,J-i}(A) \Hus_J^{J-i,J}(B)\circ q.
\end{align}
Thus \eqref{eqtn:HJ_multiplicativity_expansion} extends to $\SU(2)$ as:
\begin{align}
        \tilde\Hus_J(\Op_J(f)\Op_J(g)) = \sum_{i=0}^{2J} \tilde\Hus_J^{J,J-i}(\Op_J(f))\tilde\Hus_J^{J-i,J}(\Op_J(g)).
\end{align}
Recall that $\iset{\xi_i^L}$ are the left-invariant vector fields on $SU(2)$ corresponding to $\iset{S_i}$. One can verify the identities
\begin{align}
    \xi_i^L ({}_J\bra{\varphi}h^{-1}Th\ket{\psi}_J) = -{}_J\bra{\varphi}h^{-1}ThS_i\ket{\psi}_J+{}_J\bra{\varphi}S_ih^{-1}Th\ket{\psi}_J,
\end{align}
from which it follows that:
\begin{align}\label{eq:left_inv_vf_raising_and_lowering}
    \xi_-^L \tilde \Hus^{J,J-i}_J(T)(h) = -\sqrt{(2J-i)(i+1)}\tilde \Hus^{J,J-i-1}_J(T)(h),\\
    \xi_+^L \tilde \Hus^{J-i,J}_J(T)(h) = \sqrt{(2J-i)(i+1)}\tilde \Hus^{J-i-1,J}_J(T)(h).\nonumber
\end{align}
Since for any $f \in L^p(S^2)$, $p \in [1 , \infty]$, we have $\| f \|_p = \| f \circ q \|_p$, an iteration of the above identities gives the bound:
\begin{align}
\bigg\|\sum_{i=m}^{2J}\!\Hus^{J,J-i}_J\!&(\Op_J(f))\!\Hus^{J-i,J}_J\!(\Op_J(g))\bigg\|_p\!\leq \!\sum_{i=m}^{2J}\norm{\tilde\Hus^{J,J-i}_J\!(\Op_J(f))\tilde\Hus^{J-i,J}_J\!(\Op_J(g))}_p\\
&=\sum_{i=0}^{2J-m}\!\frac{(2J\!-\!i\!-\!m)!i!}{(2J\!-\!i)!(i\!+\!m)!}\norm{(\xi_-^L)^m\tilde\Hus^{J,J-i}_J\!(\Op_J(f))(\xi_+^L)^m\tilde\Hus^{J-i,J}_J\!(\Op_J(g))}_p.\nonumber
\end{align}
We bound each norm using the Hölder inequality for $1/p=1/p_1+1/p_2$ by:
\begin{align}
\norm{(\xi_-^L)^m\tilde\Hus^{J,J-i}_J\!(\Op_J(f))}_{p_1}\norm{(\xi_+^L)^m\tilde\Hus^{J-i,J}_J\!(\Op_J(g))}_{p_2},
\label{eq:Hus_exp_Holder}
\end{align}
where the norms are still taken over $\SU(2)$. We have $\xi_i^R \widetilde \Hus_J^{a,b}(T) =  \widetilde \Hus_J^{a,b}([S_i,T])$. One can express $\xi_i^L$ as linear combinations of right invariant vector fields with coefficients which are smooth functions on $\SU(2)$. Therefore, for example the first factor in \eqref{eq:Hus_exp_Holder} can be bounded by a sum of terms of the form
\begin{equation}
    c \| \xi_{i_1}^R \cdots \xi_{i_\ell}^R \tilde\Hus^{J,J-i}_J\!(\Op_J(f)) \|_{p_1} = c \| \tilde\Hus^{J,J-i}_J\!(\Op_J( L_{i_1} \cdots L_{i_\ell}f)) \|_{p_1},
\end{equation}
with $\ell \leq m$. Hence using \eqref{eq:Op_bound} and the analogous bound of \eqref{eq:off_Hus_bound} for $\widetilde \Hus^{a,b}$ we obtain
\begin{equation}
    \norm{(\xi_-^L)^m\tilde\Hus^{J,J-i}_J\!(\Op_J(f))}_{p_1} \leq C_m \| f \|_{m,p_1},
\end{equation}
and similarly for the second factor in \eqref{eq:Hus_exp_Holder}. 
We obtain the bound:
\begin{align}
& \bigg\|\sum_{i=m}^{2J}\!\Hus^{J,J-i}_J\!(\Op_J(f))\!\Hus^{J-i,J}_J\!(\Op_J(g))\bigg\|_p  \\
\leq & C_m^2\,\norm{f}_{m,p_1}\norm{g}_{m,p_2}\sum_{i=0}^{2J-m}\frac{(2J-i-m)!i!}{(2J-i)!(i+m)!}. \nonumber
\end{align}
Changing indices we upper bound the sum by:
\begin{align}
    & \sum_{i=0}^{2J-m}\frac{(2J-i-m)!i!}{(2J-i)!(i+m)!}\leq 2\sum_{i=0}^{\lceil J-\frac{m}{2}\rceil}\frac{(2J-i-m)!i!}{(2J-i)!(i+m)!} \\
    & \leq \frac{2}{(J\!-\!1)\!\ldots\! (J\!-\!m)}\sum_{i=0}^{\lceil J-\tfrac{m}{2}\rceil} \frac{i!}{(i+m)!}. \nonumber
\end{align}
Bounding the remaining finite sum with the series $\sum_{i=0}^\infty$, which converges, we complete the proof of the first claim.

We next want to prove the expansions \eqref{eq:Husimi_product_expansion_first_order} and \eqref{eq:Husimi_product_expansion_second_order}. We write:
\begin{align}
\label{eq:random_label_17629}
    \Hus_J(\Op_J(f)\Op_J(g)) &= \Hus_J(\Op_J(f))\Hus_J(\Op_J(g))+\Hus_J^{J,J-1}(\Op_J(f))\Hus_J^{J-1,J}(\Op_J(g))  \nonumber \\&\qquad+ \sum_{m=2}^{2J}\Hus_J^{J,J-m}(\Op_J(f))\Hus_J^{J-m,J}(\Op_J(g)).
\end{align}
We have just proved that the last sum satisfies
\begin{align}
    \norm{\sum_{m=2}^{2J}\Hus_J^{J,J-m}(\Op_J(f))\Hus_J^{J-m,J}(\Op_J(g))}_p\leq \frac{C}{(2J+1)^2}\norm{f}_{2,p_1}\norm{g}_{2,p_2}.
\end{align}
To get $\eqref{eq:Husimi_product_expansion_first_order}$, by an analogous argument,
\begin{align}
    \norm{\Hus_J^{J,J-1}(\Op_J(f))\Hus_J^{J-1,J}(\Op_J(g))}_{p}\leq \frac{C}{2J+1}\norm{f}_{1,p_1}\norm{g}_{1,p_2}.
\end{align}
Left is the first term of \eqref{eq:random_label_17629}, for which we write $\Hus_J\Op_J = 1+ (\Hus_J\Op_J-1)$:
\begin{align}
    \Hus_J(\Op_J(f))\Hus_J(\Op_J(g)) &= fg + [f(\Hus_J\Op_J-1)g+g(\Hus_J\Op_J-1)f]\\
    &\qquad+(\Hus_J\Op_J-1)f(\Hus_J\Op_J-1)g \nonumber.
\end{align}
Since $\norm{(\Hus_J\Op_J-1)f}_{p_1}\leq C(2J+1)^{-1}\norm{f}_{2,p_1}$ by \Cref{prop:HJ_OpJ_approximate_inversion_all_p}, and likewise for $g$, the claimed expansion follows.

In order to prove the expansion of \eqref{eq:Husimi_product_expansion_second_order} we write $\Hus_J\Op_J = 1+ \frac{\Delta}{2J+1}+\mathcal{E'}$ and expand the first term of \eqref{eq:random_label_17629} as:
\begin{align}
\Hus_J(\Op_J(f))\Hus_J(\Op_J(g)) = & fg + \frac{1}{2J+1}(g\Delta f+f\Delta g)+\frac{1}{(2J+1)^2}(\Delta f\Delta g) \\&+ \mathcal{E}'f\left(g+\frac{\Delta g}{2J+1}\right)+\mathcal{E}'g\left(f+\frac{\Delta f}{2J+1}\right)+ \mathcal{E}'f\mathcal{E}'g.\nonumber
\end{align}
Using the bound $\norm{\mathcal{E}'(f)}_2\leq \frac{1}{2J+1}\norm{(-\Delta)^2f}_{2}$ of \Cref{prop:HJ_OpJ_approximate_inversion_2_norm}, all but the first two terms are bounded as needed. 

We also need to expand the second term of \eqref{eq:random_label_17629}. We start by recognizing:
\begin{align}
 \Hus_J^{J,J-1}(\Op_J(f))\Hus_J^{J-1,J}(\Op_J(g))(\uparrow)& = -\frac{1}{2J} \xi_-^L \widetilde \Hus_J \Op_J(f) (e) \xi_+^L \widetilde \Hus_J \Op_J(g)(e) \label{eq:Husimi_shifted_north} \\
 = & -\frac{1}{2J} L_- \Hus_J\Op_J( f)(\uparrow) L_+ \Hus_J\Op_J( g)(\uparrow), \nonumber
\end{align}
where we used the fact that $\xi_i^L = \xi_i^R$ at the neutral element. Next, we note that for two functions $F,G$ on $S^2$ we have 
\begin{align}
     (L_- F \cdot  L_+ G)(\uparrow)  = (- \nabla F\cdot \nabla G - i\poisson{F}{G} )(\uparrow).
\end{align}
Using this in \eqref{eq:Husimi_shifted_north}, with $F = \Hus_J \Op_J f$ and $G = \Hus_J \Op_J g$, we obtain 
\begin{align}
& \Hus_J^{J,J-1}(\Op_J(f))\Hus_J^{J-1,J}(\Op_J(g))(\uparrow) \label{eq:Husimi_shifted_north2} \\
=&  \frac{1}{2J} ( \nabla \Hus_J \Op_J f \cdot \nabla \Hus_J \Op_J g+  i \poisson{\Hus_J \Op_J f}{ \Hus_J \Op_J g})(\uparrow).   \nonumber
\end{align}
By the symmetry, \eqref{eq:Husimi_shifted_north2} holds everywhere on $S^2$, not only at the north pole. Now we use \Cref{prop:HJ_OpJ_approximate_inversion_all_p} on each term and replace $\frac{1}{2J}$ by $\frac{1}{2J+1}$ at the cost of introducing an error term of order $(2J+1)^{-2}$. We get:
\begin{align}
    \Bigg\|\frac{1}{2J}\Hus_J^{J,J-1}(\Op_J(f))\Hus_J^{J-1,J}(\Op_J(g))&-\frac{1}{2J\!+\!1}(\nabla f \!\cdot\! \nabla g +i\poisson{f}{g})\Bigg\|_1\\
    &\leq\! \frac{C}{(2J\!+\!1)^2}\norm{f}_{3,2}\norm{g}_{3,2}. \nonumber
\end{align}
Combining everything, we have:
\begin{align}
    \Hus_J(\Op_J(f)\Op_J(g)) = fg + \frac{1}{2J+1}(f\Delta g + g\Delta f +\nabla f \cdot \nabla g + i\poisson{f}{g}) + \mathcal{E}_2(f,g).
\end{align}
Of course $f\Delta g + g\Delta f + \nabla f \cdot \nabla g = \Delta(fg)-\nabla f\cdot \nabla g$.
\end{proof}

\begin{proposition}\label{proposition:multiplicativity_OpJ_zeroth_order_p_bound}
If $J\in\tfrac{1}{2}\N$ and $1\leq p,p_1,p_2 \leq \infty$ with $1/p=1/p_1+1/p_2$, then:
\begin{align}
    \norm{\Op_J(fg)-\Op_J(f)\Op_J(g)}_p\leq \frac{C}{2J+1}\norm{\Op_J}_{p\to p}\norm{f}_{2,p_1}\norm{g}_{2,p_2},
    \label{eq:mult_bound}
\end{align}
for all $f\in W^{2,p_1}(S^2)$, $g\in W^{2,p_2}(S^2)$. In particular, if $p_1 \in [1 , \infty)$ and $f \in L^{p_1}(S^2)$ or $p_1 = \infty$ and $f \in C(S^2)$, and likewise for $p_2$ and $g$, then
\begin{equation}
    (2J+1)^{- \frac{1}{p}} \| \Op_J(fg) - \Op_J(f) \Op_J(g) \|_p \to 0.
    \label{eq:mult_conv}
\end{equation}
\end{proposition}

\begin{proof}
    This is almost immediate from the previous discussion. Since:
    \begin{align}\label{eqtn:OpJfOpJg_two_term_expansion}
        \Op_J(f)\Op_J(g) = (\Op_J\Hus_J)(\Op_J(f)\Op_J(g)) + (1-\Op_J\Hus_J)(\Op_J(f)\Op_J(g)),    
    \end{align}
    we have by \Cref{lemma:husimi_of_OpJfOpJg_bound} the bound:
    \begin{align}
        \norm{\Op_J\Hus_J(\Op_J(f)\Op_J(g))-\Op_J(fg)}_p\leq \frac{C}{2J+1}\norm{\Op_J}_{p\to p}\norm{f}_{2,p_1}\norm{g}_{2,p_2}.
    \end{align}
    We are done if we show that the second term of $\eqref{eqtn:OpJfOpJg_two_term_expansion}$ is bounded by a similar quantity. A computation using the Leibniz property of commutators shows that:
    \begin{align}
        (1+4\mathcal{Q})(\Op_J(f)\Op_J(g)) = & \Op_J(f)\Op_J(g) - 4\Op_J(\Delta f)\Op_J(g)-4\Op_J(f)\Op_J(\Delta g) \nonumber \\& + 8\sum_{i=1}^3 \Op_J(L_i f)\Op_J(L_i g),
    \end{align}
    so that, by \Cref{prop:HJ_OpJ_approximate_inversion_all_p} and an application of the Hölder inequality:
    \begin{align}
        \norm{(1-\Hus_J\Op_J)(\Op_J(f)\Op_J(g))}_p&\leq \frac{C}{2J+1}\norm{(1+4\mathcal{Q})(\Op_J(f)\Op_J(g))}_p\\
        &\leq \frac{C}{2J+1}\norm{\Op_J}_{p_1\to p_1}\norm{\Op_J}_{p_2\to p_2}\norm{f}_{2,p_1}\norm{g}_{2,p_2} \nonumber \\
        &=\frac{C}{2J+1}\norm{\Op_J}_{p\to p}\norm{f}_{2,p_1}\norm{g}_{2,p_2}. \nonumber 
    \end{align}
    This concludes the proof of \eqref{eq:mult_bound}. Then \eqref{eq:mult_conv} follows by a standard density argument. 
\end{proof}

\begin{proposition}\label{lemma:multiplicativity_holder_continuous_functions}
Let $f\in C^{0,\alpha}(S^2)$, $g\in C^{0,\beta}(S^2)$ be Hölder continuous, $\alpha,\beta \in [0,1]$. Then
\begin{align}
    \bigg\|\Op_J(fg)- \frac{1}{2} (\Op_J(f)\Op_J(g) &+ \Op_J(g)\Op_J(f))\bigg\|_1 \\
  & \leq C {(2J+1)^{1-\alpha/2-\beta/2}}\norm{f}_{C_h^{0,\alpha}}\norm{g}_{C^{0,\alpha}_h}.  \nonumber
\end{align}
\end{proposition}

\begin{proof}
We write
\begin{align}
    & \Op_J(fg)- \frac{1}{2} (\Op_J(f)\Op_J(g) + \Op_J(g)\Op_J(f))  \\
     &=  \frac{1}{2} \left( \Op_J(fg) \Op_J(1) + \Op_J(1) \Op_J (fg) - \Op_J(f) \Op_J(g) - \Op_J(g) \Op_J(f) \right). \nonumber
\end{align}
Expressing each $\Op(\cdot)$ as an integral over the sphere we see that this equals
\begin{align}
      & \frac{(2J+1)^2}{2} \int_{S^2 \times S^2} (f(\omega)-f(\omega'))(g(\omega) - g(\omega')) {}_J\vert{\omega}\rangle_{J\,J} \langle{\omega}\vert{\omega'}\rangle_{J\,J} \langle{\omega'}\vert \domega \domega' . \nonumber
\end{align}
We use triangle inequality and $\| \ket{\omega}_J {}_J \bra{\omega'} \|_1=1$ to get
\begin{align}
    & \| \Op_J(fg)- \frac{1}{2} (\Op_J(f)\Op_J(g) + \Op_J(g)\Op_J(f)) \|_1  \label{eq:random_numer_412} \\
     \leq &  \frac{(2J+1)^2}{2} \int_{S^2 \times S^2} |f(\omega)-f(\omega')||g(\omega) - g(\omega')| |{}_J \langle{\omega}\vert{\omega'}\rangle{}_{J} |  \domega \domega' . \nonumber
\end{align}
Let $\theta$ be the angle between $\omega$ and $\omega'$, so that $|\omega-\omega'| = \sqrt{2(1-\cos(\theta))}$. Using the Hölder's condition, we bound the right hand side of \eqref{eq:random_numer_412} by 
\begin{align}
2^{\alpha/2+\beta/2-1}(2J+1)^2 &\norm{f}_{C_h^{0,\alpha}}\norm{g}_{C^{0,\beta}_h}\int_{S^2\times S^2} (1-\cos(\theta))^{\alpha/2+\beta/2}\Big(\frac{1+\cos(\theta)}{2}\Big)^{J}\domega\domega' 
\nonumber \\
&=2^{\alpha+\beta-2}(2J+1)^2 \int_0^1 t^{\alpha/2+\beta/2}(1-t)^{J}\dt\\
&=2^{\alpha+\beta-2}(2J+1)^2 B(J+1, \alpha/2+\beta/2+1). \nonumber
\end{align}
We have $B(J+1, \alpha/2+\beta/2+1) = O(J^{-1-\alpha/2-\beta/2})$ as $J\to\infty$, finishing the proof.
\end{proof}

\begin{proposition}
For any $J\in\tfrac{1}{2}\N$ and $f,g\in W^{4,2}(S^2)$ we have the expansion
\begin{align}
\Op_J(f)\Op_J(g) = \Op_J \left(fg+\frac{i\poisson{f}{g}-\nabla f\cdot \nabla g}{2J+1} \right) + \mathcal{E}(f,g),
\end{align}
where $\norm{\mathcal{E}(f,g)}_1\leq C(2J+1)^{-2}\norm{\Op_J}_{1\to 1}\norm{f}_{{4,2}}\norm{g}_{{4,2}}$. 
\end{proposition}
\begin{proof}
We write:
\begin{align}
\Op_J(f)\Op_J(g) = \Op_J(\Hus_J(\Op_J(f)\Op_J(g))) + (1-\Op_J\Hus_J)(\Op_J(f)\Op_J(g)). \label{eq:random_equation_237}
\end{align}
It follows immediately from \Cref{lemma:husimi_of_OpJfOpJg_bound} that:
\begin{align}
    \Op_J(\Hus_J(\Op_J(f)\Op_J(g))) &= \Op_J(fg) + \frac{1}{2J+1}\Op_J(\Delta(fg)-\nabla f \cdot \nabla g +i\poisson{f}{g})\nonumber \\
    &\quad + \Op_J(\mathcal{E}_2(f,g)),
\end{align}
where $\norm{\Op_J(\mathcal{E}_2(f,g))}_1\leq C(2J+1)^{-2}\norm{\Op_J}_{1\to 1}\norm{f}_{4,2}\norm{g}_{4,2}$. 

Next we rewrite the second term in \eqref{eq:random_equation_237} as:
\begin{align}\label{eq:OpJf_OpJg_second_term}
    (1-\Op_J\Hus_J)(\Op_J(f)\Op_J(g)) &= (1-\Op_J\Hus_J)(\Op_J(fg)) \\&\quad+ (1-\Op_J\Hus_J)(\Op_J(f)\Op_J(g)-\Op_J(fg)). \nonumber
\end{align}
We expand:
\begin{align}
(1-\Op_J\Hus_J)(\Op_J(fg)) =&\Op_J\bigg(\bigg(1-\Hus_J\Op_J+\frac{\Delta}{2J+1}\bigg)(fg)\bigg) \\ &- \frac{1}{2J+1}\Op_J(\Delta(fg)), \nonumber
\end{align}
and we use \Cref{prop:HJ_OpJ_approximate_inversion_2_norm} to bound the first part. On the other hand, the second term of \eqref{eq:OpJf_OpJg_second_term} is bounded by
\begin{align}
& \bigg\|(1-\Op_J\Hus_J)(\Op_J(f)\Op_J(g)-\Op_J(fg))\bigg\|_1\\
&\leq \frac{C}{2J+1}\norm{(1+4\mathcal{Q})(\Op_J(fg)-\Op_J(f)\Op_J(g))}_1. \nonumber
\end{align}
We use the fact that:
\begin{align}
    \mathcal{Q}(\Op_J(fg)-\Op_J(f)\Op_J(g)) = & \Op_J((-\Delta f)g)-\Op_J(-\Delta f)\Op_J(g) \\ 
    &+  \Op_J(f(-\Delta g))-\Op_J(f)\Op_J(-\Delta g) \nonumber \\ 
    &+ 2\sum_{i=1}^3 \big(\Op_J(L_ifL_ig)\!-\!\Op_J(L_if)\Op_J(L_ig)\big), \nonumber
\end{align}
which is a consequence of the Leibniz property of the commutators. Carrying out the computation and using the triangle inequality we need to bound terms of the form:
\begin{align}
    \norm{\Op_J(Af Bg)-\Op_J(Af)\Op_J(Bg)}_1,
\end{align}
where $A$ and $B$ are combinations of up to two of $\iset{L_x,L_y,L_z}$. This follows from \Cref{proposition:multiplicativity_OpJ_zeroth_order_p_bound} with $p=1$, $p_1=p_2=2$:
\begin{align}    
\norm{\Op_J(Af Bg)-\Op_J(Af)\Op_J(Bg)}_1&\leq \frac{C}{2J+1}\norm{\Op_J}_{1\to 1}\norm{Af}_{2,2}\norm{Bg}_{2,2} \\
&\leq \frac{C}{2J+1}\norm{\Op_J}_{1\to 1}\norm{f}_{4,2}\norm{g}_{4,2}.\nonumber
\end{align}
Absorbing all the terms which we have bounded into the error term $\mathcal{E}(f,g)$, we have:
\begin{align}
    \Op_J(f)\Op_J(g) &= \Op_J(fg)+\frac{1}{2J+1}\Op_J(\Delta(fg)\!-\!\nabla f\! \cdot \!\nabla g \!+\! i\poisson{f}{g}\! -\! \Delta(fg)) + \mathcal{E}(f,g)\nonumber\\
    &=\Op_J(fg)+\frac{1}{2J+1}\Op_J(-\nabla f\!\cdot\! \nabla g +i\poisson{f}{g} ) + \mathcal{E}(f,g),
\end{align}
and the bound on $\mathcal{E}(f,g)$ is of the correct form. This concludes the proof.
\end{proof}

The following argument is essentially taken from \cite{laptev_szego_1996}.

\begin{proposition}\label{prop:C2_trace_bound_Berezin_Lieb}
Let $f\in W^{1,2}(S^2)$ be a real-valued function with image in a (possibly unbounded) interval $I$, and let $\varphi$ be a real-valued function on $I$ with second derivative in $L^\infty(I)$. Then:
\begin{align}
    \abs{\Tr_J \left[\varphi(\Op_J(f))-\Op_J(\varphi(f))\right]}\leq \frac{1}{2}\norm{\varphi''}_\infty \norm{\nabla f}_2^2. \label{eq:bound_2nd_derivative}
\end{align}
\end{proposition}

\begin{proof}
This is a consequence of the upper bound in the Berezin-Lieb inequality. Formulated in terms of $\Op_J$, it states that for real valued $f$ and convex $\psi$:
\begin{align}
\Tr_J[\psi(\Op_J(f))]\leq \Tr_J[\Op_J(\psi(f))].
\end{align}
We apply this to the convex functions $\psi_{\pm}(\lambda) = \norm{\varphi''}_\infty \frac{1}{2} \lambda^2 \pm \varphi(\lambda)$, obtaining:
\begin{align}\label{eq:upper_bound_C2_from_berezin_lieb}
    \pm\Tr_J[\varphi(\Op_J(f))-\Op_J(\varphi(f))]\leq \frac{\norm{\varphi''}_\infty}{2}\Tr_J[\Op_J(f^2)-\Op_J(f)^2].
\end{align}
We can replace the left hand side by its absolute value. We bound the right hand side by first rewriting:
\begin{align}
    \Tr_J[\Op_J(f^2)-\Op_J(f)^2] &= (2J+1)^2\int_{S^2\times S^2} \!\!\!\!\!\!\!(f(\omega)^2-f(\omega)f(\omega')) \Tr_J[P_J(\omega)P_J(\omega')]\domega \domega' \nonumber \\
    &=(2J+1)\int_{S^2}f(\omega)(1-\Hus_J\Op_J)f(\omega)\domega\\
    &=(2J+1)\norm{(1-\Hus_J\Op_J)^{1/2}f}_2^2, \nonumber
\end{align}
and by \Cref{prop:HJ_OpJ_approximate_inversion_2_norm} we get the upper bound:
\begin{align}
    \Tr_J[\Op_J(f^2)-\Op_J(f)^2]\leq 
 \| (-\Delta)^{\frac{1}{2}} f \|_2^2= \norm{\nabla f}_2^2,
\end{align}
which together with \eqref{eq:upper_bound_C2_from_berezin_lieb} gives the desired result.
\end{proof}

Of course, the above also works for complex-valued $\varphi$, by separately bounding the real and imaginary parts. Then we have to remove the factor $\frac12$ on the right hand side of \eqref{eq:bound_2nd_derivative}.

If $\varphi$ is a convex function, the Berezin-Lieb inequalities give us a good bound on the speed of convergence under a weaker regularity assumption, namely that $\varphi$ is H\"older continuous. 

\begin{proposition} \label{prop:convex_Holder_trace}
    Let $f\in W^{2,1}(S^2)$ be real valued, and let $\varphi \in C^{0,\alpha}(\image(f))$ be a real valued convex function, $\alpha\in (0,1]$. Then for every $J\in \tfrac{1}{2}\N$ we have:
    \begin{align}
        0\leq \frac{1}{2J+1}\Tr_J[\Op_J(\varphi(f))-\varphi(\Op_J(f))]\leq \frac{C}{(2J+1)^{\alpha}}\norm{\varphi}_{C^{0,\alpha}_h}\norm{(1-4\Delta)f}_{1}^\alpha.
    \end{align}
\end{proposition}
\begin{proof}
    By the Berezin-Lieb inequality we know that:
    \begin{align}
        0 & \leq  \frac{1}{2J+1}\Tr_J[\Op_J(\varphi(f))-\varphi(\Op_J(f))] \\  & \leq \int_{S^2} 
\left[ \varphi(f(\omega))-\varphi(\Hus_J\Op_J f(\omega)) \right] \domega  \nonumber \\
& \leq \norm{\varphi}_{C^{0,\alpha}_h}\int_{S^2}\abs{(1-\Hus_J\Op_J)f}^\alpha\domega. \nonumber
    \end{align}
    By Jensen's inequality and \Cref{prop:HJ_OpJ_approximate_inversion_all_p}:
    \begin{align}
\norm{\varphi}_{C^{0,\alpha}_h}\int_{S^2}\abs{(1-\Hus_J\Op_J)f}^\alpha\domega 
& \leq \norm{\varphi}_{C^{0,\alpha}_h}\norm{(1-\Hus_J\Op_J)f}^\alpha_1 \\
& \leq \norm{\varphi}_{C^{0,\alpha}_h} \left( \frac{C}{2J+1} \right)^\alpha \| (1-4\Delta)f \|_{1}^\alpha. \nonumber
    \end{align}
    Now we use $\sup_{\alpha \in [0,1]}  C^\alpha < \infty$.
\end{proof}

\begin{corollary}\label{lemma:OpJ_functional_calculus_specialized_to_entropy}
    Let $\varphi(x) = x\log(x)$. For $f\in W^{2,1}(S^2)$ real valued, with $0\leq f \leq \Lambda$, and $J\in \tfrac{1}{2}\N$ with $J \geq 1$:
    \begin{align}\label{eq:OpJ_functional_calculus_specialized_to_entropy}
    0 \leq & \int_{S^2} \varphi ( f (\omega)) \domega -  \frac{1}{2J+1}\Tr_J[\Op_J(\varphi(f))]    \\
     \leq &  C \log(2J+1)\norm{(1-4\Delta)f}^{1-\frac{1}{\log(2J+1)}}_{1}+{C}\max\{0 ,  1+\log(\Lambda) \} \norm{(1-4\Delta)f}_{1}. \nonumber
    \end{align}
\end{corollary}

\begin{proof}
    We write $\varphi(x) = \varphi_1(x)+\varphi_2(x)$, where:
    \begin{align}
        \varphi_1(x) = \begin{cases}\varphi(x),\quad &x<1/e\\-1/e,&x\geq 1/e\end{cases},\quad\textup{and}\quad \varphi_2(x)=\begin{cases}0,\quad &x<1/e\\\varphi(x)+1/e,&x\geq 1/e\end{cases}.
    \end{align}
    Functions $\varphi_1, \varphi_2$ are convex, with $\varphi_1$ non-increasing and $\varphi_2$ non-decreasing. We estimate the Lipschitz seminorm of $\varphi_2$ using the derivative:
    \begin{align}
\norm{\varphi_2}_{C^{0,1}_h}= \sup_{x\in[0,\Lambda]}{\varphi_2'} = \max \{0, 1+\log(\Lambda) \}.
    \end{align}
Next we estimate Hölder seminorms of $\varphi_1$. We can restrict attention to the interval $[0,1/e]$. On this interval, one can show that $\abs{\varphi_1'(x)/(x^\alpha)'}\leq 1/\alpha(1-\alpha)$. In particular then, for $x\leq y$:
    \begin{align}
        \abs{\varphi_1(y)-\varphi_1(x)} & \leq \int_{x}^y |\varphi_1'(t)|\dt \leq \frac{1}{\alpha(1-\alpha)}\int_x^y \alpha t^{\alpha-1}\dt \\
        & \leq \frac{y^\alpha-x^\alpha}{\alpha(1-\alpha)}\leq \frac{\abs{y-x}^\alpha}{\alpha(1-\alpha)} . \nonumber
    \end{align}
    This gives $\norm{\varphi_1}_{C^{0,\alpha}_h}\leq \tfrac{1}{\alpha(1-\alpha)}$.

    \Cref{prop:convex_Holder_trace} applied separately to $\varphi_1$ and $\varphi_2$ now gives:
    \begin{align}
        \frac{1}{2J+1}\Tr_J[\Op_J(\varphi(f))-\varphi(\Op_J(f))]&\leq \frac{C}{(2J+1)^{\alpha}} \frac{1}{\alpha(1-\alpha)}\norm{(1-4\Delta)f}_{1}^\alpha \\
        &\quad + \frac{C}{2J+1}
\max \{0, 1+\log(\Lambda) \}
\norm{(1-4\Delta)f}_{1}. \nonumber
    \end{align}
    We choose $\alpha = 1-1/\log(2J+1) \geq 1 - \frac{1}{\log(3)} > 0$ to optimize dependence on $J$ and find that:
    \begin{align}
        \frac{1}{2J+1}\Tr_J[\Op_J(\varphi(f))\!-\!\varphi(\Op_J(f))]&\leq C\frac{\log(2J+1)}{2J+1}\norm{(1-4\Delta)f}_1^{1-\frac{1}{\log(2J+1)}}\\
        &\quad+\frac{C}{2J+1}
   \max \{0, 1+\log(\Lambda) \}
    \norm{(1-4\Delta)f}_{1}. \nonumber
    \end{align}
\end{proof}

\begin{corollary}
Let $f\in C(S^2)$ be real valued and let $\varphi \in C(\image f)$. Then
\begin{equation}
\lim_{J \to \infty}    \frac{1}{2J+1}\Tr_J[\varphi(\Op_J(f))] = \int_{S^2}\varphi(f(\omega))\domega.
\end{equation}
Equivalently, for each $J\in \tfrac{1}{2}\N$ consider the probability measure
\begin{equation}
    \mu_{J,f}:=\frac{1}{2J+1}\sum_{i=1}^{2J+1}\delta_{\lambda_i(\Op_J(f))},
\end{equation}
where $\lambda_i(\Op_J(f))$ is the $i$-th eigenvalue of $\Op_J(f)$. Then $\mu_{J,f}$ converges in the weak-$*$ topology to the pushforward measure $\domega \circ f^{-1}$ in the limit $J \to \infty$.
\end{corollary}
\begin{proof}
    Follows from \Cref{prop:C2_trace_bound_Berezin_Lieb} by a standard density argument.
\end{proof}

\section{Application to channels} \label{sec:channels}

Due to the multiplicity-free nature of the decomposition of the tensor product of two representations, the set of $\SU(2)$-equivariant quantum channels with fixed input and output representations $\irrep_J$ and $\irrep_K$ is a simplex. This result is originally due to \cite{nuwairan_su2-irreducibly_2013}. In~terms of our notation:
\begin{proposition}[\cite{nuwairan_su2-irreducibly_2013}] \label{prop:simplex}
Equivariant quantum channels from $\irrep_J$ to $\irrep_K$ form a simplex with vertices $\Phi_{J,K}^M$, where
\begin{align}
    \Phi^{M}_{J,K} (\rho) = \frac{2J+1}{2M+1}\Tr_J [P^M_{J,K}(\rho^\beta \tensor \identity_K)],
\label{eq:channel_formula_P}
\end{align}
for $M\in\iset{\abs{K-J},\ldots, K+J}$. That is, equivariant quantum channels from $\irrep_J$ to $\irrep_K$ are of the form
\begin{align}
    \Phi_{J,K}^{(\lambda)} := \sum_{M=\iabs{K-J}}^{K+J}\lambda_M\Phi^M_{J,K},
\end{align}
with uniquely determined $\lambda_M \geq 0$ satisfying $\sum_M \lambda_M =1$.
\end{proposition}

See \Cref{Appendix:QC} for a proof. We remark that $\rho^\beta$ appears in \eqref{eq:channel_formula_P} because, when applying \eqref{eq:channel_from_Choi}, we use $\beta$ to identify representations of $\SU(2)$ with their dual spaces. Two equivalent formulas for $\Phi^M_{J,K}$ were given in \eqref{eq:channel_formula_2} and \eqref{eq:channel_formula_3}.

The adjoint $\Phi^*$ of a map $\Phi : B(\irrep_J) \to B(\irrep_K)$ is defined by
\begin{equation}
    \Tr_{\irrep_K}(T^* \Phi(S)) = \Tr_{\irrep_J}(\Phi^*(T)^* S). 
\end{equation}
We note that
\begin{equation}
    (\Phi_{J,K}^M)^* = \frac{2J+1}{2K+1} \Phi_{K,J}^M.
\end{equation}

The argument below is essentially taken from \cite{perez-garcia_contractivity_2006}.

\begin{lemma}
For any $\SU(2)$-equivariant quantum channel $\Phi\colon B(\irrep_J)\to B(\irrep_K)$:
\begin{align}\label{eqtn:channel_p_p_norm}
    \norm{\Phi}_{p\to p} = \left(\frac{2J+1}{2K+1}\right)^{1-1/p},\quad p\in [1,\infty].
\end{align}
\end{lemma}
\begin{proof}
It is known \cite[Corollary 1]{russo_note_1966} that $\inorm{\Phi}_{\infty\to \infty}=\inorm{\Phi(\identity)}_\infty$ for any positive map $\Phi$, and in our case $\Phi(\identity)=\tfrac{2J+1}{2K+1} \identity$. On the other hand,
\begin{align}
    \norm{\Phi}_{1\to 1} = \sup_{\| S \|_\infty = \| T \|_1 =1} \left| \Tr \left( S \Phi(T) \right) \right|  = \sup_{\| S \|_\infty = \| T \|_1 =1}\left| \Tr \left( \Phi^*(S) T \right) \right| = \norm{\Phi^*}_{\infty\to\infty}.
\end{align}
The adjoint of a channel is a unital map, so $\norm{\Phi^*}_{\infty\to\infty}=1$. By interpolation we recover \eqref{eqtn:channel_p_p_norm} as an upper bound. Calculating $\| \Phi(\identity) \|_p$ one finds that the upper bound is attained for every $p$.
\end{proof}

Our strategy is to replace $\Phi(\rho)$ with the quantization of some function. The following lemma identifies the correct function.

\begin{proposition}\label{theorem:convergence_of_husimi}
For every $\rho \in B(\irrep_J)$ and $p\in[1,\infty]$ we have:
\begin{align}
(2J+1)^{1/p}\norm{\frac{2K+1}{2J+1}\Hus_K\Phi^{K+i}_{J,K}(\rho)-\Hus^{-i}_J(\rho)}_{p}\leq 2\frac{(J+i)(J-i+1)}{2K-J+i+1}\norm{\rho}_p.
\label{eq:Hus_of_Phi_bound}
\end{align}

In particular, if $2J \leq K$ and $\lambda = \iset{\lambda_i}_i$ satisfy $\lambda_i\geq 0$, $\sum_i \lambda_i=1$, we have:
\begin{align}
    (2J+1)^{1/p}\norm{\frac{2K+1}{2J+1}\Hus_K(\Phi^{(\lambda)}_{J,K}(\rho))-\Hus^{(\lambda)}_J(\rho)}_{p}\leq \frac{(2J+1)^{2}}{2K+1}  \norm{\rho}_p.
    \label{eq:Hus_of_Phi_bound_2}
    \end{align}

\end{proposition}
\begin{proof}
By the definition of the dual channel:
\begin{align}
    \frac{2K+1}{2J+1}\Hus_K[\Phi^{K+i}_{J,K}(\rho)](\omega) &= \Tr_{J}[\rho \,\Phi^{K+i}_{K,J}(\ket{\omega}\bra{\omega})] \\
    &=\Tr_J[\rho\Tr_{K+i}[\iota_{J,K+i}^K \ket{\omega}\bra{\omega}q_{J,K+i}^K]]. \nonumber
\end{align}
We have $\iota^K_{J,K+i}{\ket{\omega}_K} = \sum_{\ell=-J}^Jc_\ell\ket{\omega;\,-i+\ell}\ket{\omega;\,K+i-\ell}$, where $c_\ell = C^{K,K}_{J, -i+\ell ; K+i,K+i- \ell}$ are the Clebsch-Gordan coefficients. This allows to compute the partial trace,
\begin{align}
\frac{2K+1}{2J+1}\Hus_K[\Phi^{K+i}_{J,K}(\rho)](\omega) &= \Tr_J[\rho \sum_{\ell}\abs{c_\ell}^2 \ket{\omega;\,-i+\ell}\bra{\omega;\,-i+\ell}]\\
&=\sum_{\ell} \abs{c_\ell}^2 \Hus_J^{-i+\ell}(\rho)(\omega).
\end{align}
We combine this with the triangle inequality to obtain:
\begin{align*}
    \norm{\frac{2K+1}{2J+1}\Hus_K(\Phi^{K+i}_{J,K}(\rho))-\Hus^{-i}_J(\rho)}_{p}&\leq \sum_{\ell\neq 0}\abs{c_\ell}^2\norm{\Hus_J^{-i+\ell}(\rho)}_p + (1-\abs{c_{0}}^2)\norm{\Hus_J^{-i}(\rho)}_p\\
    &\leq \Big[\sum_{\ell\neq 0}\abs{c_\ell}^2+(1-\abs{c_{0}}^2)\Big]\left(2J+1\right)^{-1/p}\norm{\rho}_p\\
    &=2(1-\abs{c_{0}}^2)\left(2J+1\right)^{-1/p}\norm{\rho}_p.
\end{align*}
Now \eqref{eq:Hus_of_Phi_bound} follows from \Cref{lem:CG_bounds}. To obtain \eqref{eq:Hus_of_Phi_bound_2} we optimize over $i$ in the numerator and lower bound the denominator: $2K-J+i+1 \geq K + \frac{1}{2}$.
\end{proof}

The next step is to understand how far the quantum channel is from the quantization of the corresponding Husimi function. 

\begin{proposition}\label{theorem:convergence_of_channels}
    For every $\rho \in B(\irrep_J)$ and $p \in [1,\infty]$ we have:
    \begin{equation}
    \left( \frac{2K+1}{2J+1} \right)^{- \frac{1}{p}} \norm{ \frac{2K+1}{2J+1} \Phi^{K+i}_{J,K}(\rho) - \Op_K \Hus_J^{-i} (\rho)  }_p \leq 12 \frac{(J-i)(J+i+1)}{2K-J+i+1} \| \rho \|_p.
    \end{equation}
    In particular, if $2J \leq K$ and $\lambda = \iset{\lambda_i}_i$ satisfy $\lambda_i\geq 0$, $\sum_i \lambda_i=1$:
    \begin{align}
        \left( \frac{2K+1}{2J+1} \right)^{- \frac{1}{p}} \norm{ \frac{2K+1}{2J+1} \Phi^{(\lambda)}_{J,K}(\rho) - \Op_K \Hus_J^{(\lambda)} (\rho)  }_p\leq 6\frac{(2J+1)^2}{2K+1}\norm{\rho}_p.
    \end{align}
\end{proposition}
\begin{proof}
Let $\epsilon = \frac{(J-i)(J+i+1)}{2K-J+i+1}$. We may assume that $\epsilon  \leq \tfrac{1}{5}$, since the bound is trivial otherwise. 
    
    We will use the off-diagonal Husimi functions ${\Hus}^{a,b}_J$ of \eqref{eq:off-Hus}, and more precisely their lifts $\widetilde{\Hus}^{a,b}_J$ to $\SU(2)$ already introduced in \eqref{eq:off_diagonal_husimi_group}. Recall that $\widetilde{\Hus}^{a,b}_J(A)$ is the function on $\SU(2)$ given by $g\mapsto {}_J\bra{a}g^{-1}Ag\ket{b}_J$. We will also use maps $\widetilde\Op_K^{a,b}$, defined as:
\begin{align}\label{eq:generalized_op_tilde_def}
        \widetilde\Op_K^{a,b}(f) &= (2K+1)\int_{\SU(2)} f(g) g\ketbra{a}{b}{K}g^{-1}\dg
    \end{align}
    for a function $f$ on SU(2). One can check that $(2K+1)\widetilde \Hus_K^{a,b} = (\widetilde\Op_K^{a,b})^*$ and \begin{align}\label{eq:norm_of_generalized_op_hus_su2}
        \norm{\widetilde\Op_K^{a,b} \widetilde\Hus_J^{c,d}}_{p\to p} \leq \left(\frac{2K+1}{2J+1}\right)^{1/p}.
    \end{align}

We can express $\Phi_{J,K}^{K+i}$ in terms of the maps $\widetilde\Op_J^{a,b}$ and $\widetilde \Hus_J^{a,b}$ 
as follows: 
    \begin{align}\label{eq:phi_expressed_as_off_diagonal_husimis}
        \Phi_{J,K}^{K+i}(\rho) &= \frac{2J+1}{2M+1}\Tr_J\left[P_{J,K}^M(\rho^\beta \tensor \identity_K)\right]\\
        &= (2J+1)\int_{S^2} \Tr_J\left[\iota_{J,K}^{K+i}\ketbra{\omega}{\omega}{K+i}q_{J,K}^{K+i}(\rho^\beta\tensor \identity_K)\right] \domega\nonumber\\
        &= \frac{2J+1}{2K+1} \sum_{\ell_1,\ell_2\geq 0} c_{\ell_1}c_{\ell_2} \widetilde \Op_K^{K-\ell_1,K-\ell_2}\widetilde \Hus_J^{i+\ell_2,i+\ell_1}(\rho^\beta),\nonumber
    \end{align}
    where $c_\ell = C^{K+i,K+i}_{J,i+\ell;K,K-\ell}$ (see \eqref{eq:CG_def}). 
    We note that:
    \begin{align}
        \widetilde \Hus_J^{a,b}(\rho^\beta) = (-1)^{2J-a-b}\widetilde \Hus^{-b,-a}(\rho),
    \end{align}
    so \cref{eq:phi_expressed_as_off_diagonal_husimis} above becomes:
    \begin{align}
        \Phi_{J,K}^{K+i} = \frac{2J+1}{2K+1} \sum_{\ell_1,\ell_2\geq 0} (-1)^{\ell_1+\ell_2}c_{\ell_1}c_{\ell_2} \widetilde \Op_K^{K-\ell_1,K-\ell_2}\widetilde \Hus_J^{-i-\ell_1,-i-\ell_2}.
    \end{align}
    We group terms in this expression as follows: \begin{align}\label{eq:channel_bound_I_II_III_IV}
        \frac{2K+1}{2J+1}\Phi_{J,K}^{K+i}(\rho) - \Op_K\Hus_J^{-i} &= (\abs{c_0}^2-1)\Op_K\Hus_J^{-i} \\
        &\quad - c_0c_1 
 \left( \widetilde\Op_K^{K-1,K}\widetilde\Hus_J^{-i-1,-i}+\widetilde\Op_K^{K,K-1}\widetilde\Hus_J^{-i,-i-1} \right) \nonumber\\
        &\quad + c_0 \sum_{\ell\geq 2} (-1)^\ell c_\ell \left( \widetilde\Op_K^{K-\ell,K}\widetilde\Hus_J^{-i-\ell,-i}+\widetilde\Op_K^{K,K-\ell}\widetilde\Hus_J^{-i,-i-\ell} \right) \nonumber\\
        &\quad + \sum_{\ell_1,\ell_2\geq 1} (-1)^{\ell_1+\ell_2} c_{\ell_1}c_{\ell_2}\widetilde\Op_K^{K-\ell_1,K-\ell_2}\widetilde\Hus_J^{-i-\ell_1,-i-\ell_2}\nonumber\\
        &=I+II+III+IV,\nonumber
    \end{align}
    where in the first line we used that $\Op_K^a \Hus_J^b = \widetilde{\Op}_K^a \widetilde\Hus_J^b $. Using \Cref{lem:CG_bounds} and \eqref{eq:norm_of_generalized_op_hus_su2} we obtain the bounds:
    \begin{align}\label{eq:bound_I_III_IV}
        \norm{I}_{p\to p} & \leq \epsilon \left(\frac{2K+1}{2J+1}\right)^{1/p},\nonumber\\
        \norm{III}_{p\to p}&\leq \frac{2\epsilon}{1-\epsilon^{1/2}}  \left(\frac{2K+1}{2J+1}\right)^{1/p}, \\
        \norm{IV}_{p\to p}&\leq \frac{\epsilon}{(1-\epsilon^{1/2})^2} \left(\frac{2K+1}{2J+1}\right)^{1/p}. \nonumber
    \end{align}
    By assumption on $\epsilon$, we can bound $\frac{1}{1-\epsilon^{\frac{1}{2}}} \leq \frac{1}{1-\left( \frac{1}{5} \right)^{\frac{1}{2}}}$.

    To see that $II$ also has the desired decay property, we note the formulas:
    \begin{align}
            \widetilde \Op^{K,K}_K \xi_+^L  &= -\sqrt{2K}\widetilde \Op^{K,K-1}_K,\\
            \widetilde \Op^{K,K}_K \xi_-^L  &= \sqrt{2K}\widetilde \Op^{K-1,K}_K,\nonumber
    \end{align}
    which follow from \eqref{eq:left_inv_vf_raising_and_lowering} by taking adjoints. It follows easily that:
    \begin{align}\label{eq:useful_identity_generalized_Op_Hus_group}
        \widetilde \Op^{K,K-1}_K\widetilde\Hus_{J}^{-i,-i-1} &= \sqrt{\frac{(J-i)(J+i+1)}{2K}}\Op_K(\Hus_J^{-i}-\Hus_J^{-i-1}),\\
        \widetilde \Op^{K-1,K}_K\widetilde\Hus_{J}^{-i-1,-i} &= \sqrt{\frac{(J-i)(J+i+1)}{2K}}\Op_K(\Hus_J^{-i}-\Hus_J^{-i-1}).\nonumber
    \end{align}
    Since the coefficient in \eqref{eq:useful_identity_generalized_Op_Hus_group} is bounded by $\epsilon^{1/2}$, we can bound $II$ using the triangle inequality by:
    \begin{align}
        \norm{II}_{p\to p}\leq 4\epsilon\left(\frac{2K+1}{2J+1}\right)^{1/p}.
    \end{align}
    This, combined with \eqref{eq:bound_I_III_IV}, gives a bound
    \begin{align}
   & \left( \frac{2K+1}{2J+1} \right)^{-1/p}    \left \| \frac{2K+1}{2J+1}\Phi_{J,K}^{K+i} - \Op_K \Hus_J^{i} \right \|_{p \to p} \\
   & \leq \left( 5 + \frac{2}{1 - \left( \frac15 \right)^{\frac12}} + \frac{1}{\left(1- \left( \frac15 \right)^{\frac12} \right)^2} \right) \epsilon \leq 12 \epsilon. \nonumber
    \end{align}
The case $2J\leq K$ is treated similarly as in the proof of \Cref{theorem:convergence_of_husimi}.
\end{proof}

Our last step is to show convergence of traces of functions of channel outputs to corresponding integrals on $S^2$. The following lemma appears to be a well known fact, but we include it for reader's convenience.

\begin{lemma}\label{lemma:trace_bound_holder_functions}
Let $\varphi \in C^{0,\alpha}([a,b])$ be H\"older continuous with $0<\alpha\leq 1$. Then for any self-adjoint $N\times N$ matrices $A$ and $B$ with $\sigma(A),\sigma(B)\subset [a,b]$:
\begin{align}
    \abs{\Tr[\varphi(A)-\varphi(B)]}\leq \norm{\varphi}_{C_h^{0,\alpha}} \norm{A-B}_1^\alpha N^{1-\alpha}.
\end{align}
\end{lemma}
\begin{proof}
Let $b\geq \lambda_1\geq \ldots \geq \lambda_N\geq a$ and $b\geq \mu_1\geq \ldots \geq \mu_N\geq a$ denote the eigenvalues of $A$ and $B$ respectively. By the H\"older condition:
\begin{align}
    \abs{\Tr[\varphi(A)-\varphi(B)]} \leq \sum_{i=1}^N \abs{\varphi(\lambda_i)-\varphi(\mu_i)} \leq \norm{\varphi}_{C_h^{0,\alpha}} \sum_{i=1}^N \abs{\lambda_i-\mu_i}^\alpha.
\end{align}
Since $x\mapsto x^\alpha$ is concave, the result now follows from Jensen's inequality, together with the bound:
\begin{align}
    \sum_{i=1}^N\abs{\lambda_i-\mu_i}\leq \norm{A-B}_1,
\end{align}
which is a consequence of an identity of Lidskii's (see \cite[Theorem 1.20]{simon_trace_2005}).
\end{proof}

\begin{proposition}\label{thm:c2_and_holder_convergence_of_traces_channels}
If $\rho$ is a density matrix, then
\begin{equation}\label{eq:trace_formula_channels_with_error_term}
    \frac{1}{2K+1}\Tr_K \left(\varphi \left(\frac{2K+1}{2J+1}\Phi_{J,K}^{K+i}(\rho)\right)\right) = \int_{S^2}\varphi(\Hus_J^{-i}(\rho))\domega + \mathcal E,
\end{equation}
where the error term $\mathcal E$ satisfies the following bounds:
\begin{enumerate}
\item If $\varphi$ has second derivative in $ L^\infty([0,1])$, then 
$
|\mathcal E| \leq 10 \norm{\varphi''}_{\infty} \frac{J-\abs{i}+1}{2K-J+i+1},
$
\item If $\varphi \in C^{0, \alpha}([0,1])$ is convex, then $|\mathcal E| \leq C \norm{\varphi}_{C_{h}^{0,\alpha}} \left(\frac{J-\abs{i}+1}{2K-J+i+1}\right)^{\alpha}$ with a~universal constant $C$.
\end{enumerate}

Similarly, if in \eqref{eq:trace_formula_channels_with_error_term} we replace $\Phi_{J,K}^{K+i}$ with $\Phi_{J,K}^{(\lambda)}$ and $\Hus_J^{-i}$ with $\Hus_J^{(\lambda)}$, then under the assumption $K \geq 2J$ the error term $\mathcal E$ can be bounded as follows:
\begin{enumerate}
\item If $\varphi$ is a function with second derivative in $ L^\infty([0,1])$, then 
$
|\mathcal E| \leq 4 \norm{\varphi''}_{\infty} \frac{2J+1}{2K+1},
$
\item If $\varphi \in C^{0, \alpha}([0,1])$ is convex, then $|\mathcal E| \leq C \norm{\varphi}_{C_{h}^{0,\alpha}} \left(\frac{2J+1}{2K+1}\right)^{\alpha}$ with a universal constant $C$.
\end{enumerate}
\end{proposition}

\begin{proof}
    We will only prove the statements for the extremal channels, since the general case is completely analogous. We consider 1. first.

    If $A$ and $B$ are self-adjoint matrices and $\Tr(A)=\Tr(B)$, then $\Tr(f(A)) = \Tr(f(B))$ for all affine functions $f$. Applying this to $A=\tfrac{2K+1}{2J+1}\Phi_{J,K}^{K+i}(\rho)$ and $B=\Op_K(\Hus_J^{-i}(\rho))$, we obtain
    \begin{equation}
        \Tr_K \left[ \varphi(A) - \varphi(B) \right] = \Tr_K \left[ \tilde \varphi(A) - \tilde \varphi(B) \right], 
    \end{equation}
    where $\tilde \varphi(x) = \varphi(x) - \varphi(\tfrac12) - (x- \tfrac12) \varphi'(\tfrac12)$. Notice that $\varphi'' = \tilde\varphi''$ and $\norm{\tilde\varphi}_{C^{0,1}_h}\leq \tfrac{1}{2}\norm{\varphi''}_\infty$, so from \Cref{lemma:trace_bound_holder_functions} and \Cref{theorem:convergence_of_channels}:
    \begin{align}\label{eq:bound_from_lipschitz_trace_formula_using_second_derivative}
        \frac{1}{2K+1}\Bigg\vert\Tr_K\bigg[\varphi\left(\frac{2K+1}{2J+1}\Phi_{J,K}^{K+i}(\rho)\right)&-\varphi\bigg(\Op_K \big(\Hus_J^{-i}(\rho) \big)\bigg)\bigg]\Bigg\vert\\
        &\leq \frac{\norm{\varphi''}_\infty}{2(2K+1)} \norm{\frac{2K+1}{2J+1}\Phi^{K+i}_{J,K}(\rho)-\Op_K(\Hus^{-i}_J(\rho))}_1\nonumber\\
        &\leq 6\norm{\varphi''}_\infty \frac{J-\abs{i}+1}{2K-J+i+1}.\nonumber
    \end{align}
    
    On the other hand, it follows from \Cref{prop:C2_trace_bound_Berezin_Lieb} that:
    \begin{align}\label{eq:bound_from_second_derivative_on_ops}
        \bigg\vert\frac{1}{2K+1}\Tr_K \Big[ \varphi \big(\Op_K(\Hus_J^{-i}(\rho)) \big) \Big]-& \int_{S^2}\varphi(\Hus_J^{-i}(\rho))(\omega)\domega\bigg\vert \\
 & \leq \frac{\norm{\varphi''}_\infty\norm{\nabla\Hus^{-i}_J(\rho)}_{2}^2}{2K+1}, \nonumber
    \end{align}
    for all density matrices $\rho$. We bound using Cauchy-Schwarz:
    \begin{equation}
        \norm{\nabla\Hus^{-i}_J(\rho)}_{2}^2 \leq \norm{\Hus^{-i}_J(\rho)}_2 \norm{ \Delta \Hus^{-i}_J(\rho)}_2,
        \label{eq:hus_cauchy_schwarz}
    \end{equation}
    and then use \eqref{eq:ad_Casimir} to get
    \begin{align}
     (-\Delta) \Hus^{-i}_J\rho = \Hus_J^{-i} \cQ \rho = & 2(J(J+1)-i^2)\Hus_J^{-i}\rho  - (J-i)(J+i+1) \Hus_J^{-i-1} \rho  \label{eq:lap_hus} \\ 
     &  - (J+i) (J-i+1) \Hus_J^{-i+1}\rho. \nonumber 
    \end{align}
    By the bound $\| \Hus_J^a\|_{2 \to 2} = (2J+1)^{-\frac12}$ and the triangle inequality, this implies
    \begin{equation}
       \norm{\Delta \Hus_J^{-i} \rho}_2 \leq  4 \left( J(J+1) - i^2 \right) (2J+1)^{-\frac12} \leq 4 (2J+1)^{\frac12}(J-|i|+1) . \label{eq:lap_hus_bound}
    \end{equation}
Plugging this into \eqref{eq:hus_cauchy_schwarz} completes the proof of 1.

    If instead $\varphi\in C^{0,\alpha}([0,1])$ is convex, \Cref{lemma:trace_bound_holder_functions} gives:
\begin{align}\label{eq:bound_from_holder_trace_formula}
    \frac{1}{2K+1}\Bigg\vert\Tr_K\!\left(\!\varphi\!\left(\frac{2K+1}{2J+1}\Phi_{J,K}^{K+i}(\rho)\right)\!\right)&-\Tr_K(\varphi(\Op_K(\Hus_J^{-i}(\rho))))\Bigg\vert \\
    &\leq \frac{\norm{\varphi}_{C_h^{0,\alpha}}}{(2K+1)^\alpha}\norm{\frac{2K+1}{2J+1}\Phi_{J,K}^{K+i}(\rho)-\Op_K(\Hus_J^{-i}(\rho))}_1^\alpha,\nonumber
\end{align}
and \Cref{theorem:convergence_of_channels} implies the desired bound on this first part. On the other hand, from \Cref{prop:convex_Holder_trace} we have:
    \begin{align}
        \Bigg\vert\frac{1}{2K+1}\Tr_K(\varphi(\Op_K(\Hus_J^{-i}(\rho))))-\int_{S^2}&\varphi(\Hus_J^{-i}(\rho))(\omega)\domega\Bigg\vert\\
        &\leq C \frac{\norm{\varphi}_{C_h^{0,\alpha}}\norm{(1-4\Delta)\Hus_J^{-i}(\rho)}_{1}^\alpha}{(2K+1)^\alpha}.\nonumber
    \end{align}
    Proceeding as in \eqref{eq:lap_hus} and \eqref{eq:lap_hus_bound}, we bound
    \begin{align}
        \norm{(1-4\Delta)\Hus_J^{-i}(\rho)}_{1}
        \leq C(J-\abs{i}+1).
    \end{align}
    This, combined with \eqref{eq:bound_from_holder_trace_formula}, is the desired result.
\end{proof}

We recover the following for a general continuous function $\varphi$.

\begin{corollary}[Traces of functions converge to classical integrals]\label{prop:weak_star_convergence_fcts_of_channels}
Fix $J\in\tfrac{1}{2}\N$ and nonnegative $\iset{\lambda_i}_{i=-J}^J$ with $\sum_{i=-J}^J\lambda_i=1$. Let $\rho\in B(\irrep_J)$ be a density matrix. 
Then for any $f\in C([0,1])$:
\begin{align}
    \frac{1}{2K+1}\Tr\left[f\left(\frac{2K+1}{2J+1}\Phi_{J,K}^{(\lambda)}(\rho)\right)\right]\to \int_{S^2} f(\Hus_J^{(\lambda)}(\rho)(\omega)) \domega\quad \text{as }K\to\infty.
\end{align}
Moreover, convergence is uniform in $\rho$ over the set of density matrices.
\end{corollary}

\begin{proof}
    This follows from \Cref{thm:c2_and_holder_convergence_of_traces_channels} by a simple approximation argument.
\end{proof}

\begin{proposition}[Quantified convergence of von Neumann entropy]
Let $J,K\in \tfrac{1}{2}\N$ with $K\geq 1$. Then for any density matrix $\rho\in B(\irrep_J)$:
\begin{align}
    S_{vN}(\Phi^{K+i}_{J,K}(\rho)) = \log\left(\frac{2K+1}{2J+1}\right)-(2J+1)\int_{S^2}\Hus_J^{-i}(\rho)\log \Hus_J^{-i}(\rho) \domega +\; \mathcal{E},
\end{align}
with $\abs{\mathcal{E}}\leq C\log(2K+1)\tfrac{(2J+1)(J-\abs{i}+1)}{2K-J+i+1}$.

If also $K\geq 2J$:
\begin{align}
    S_{vN}(\Phi^{(\lambda)}_{J,K}(\rho)) = \log\left(\frac{2K+1}{2J+1}\right)-(2J+1)\int_{S^2}\Hus_J^{(\lambda)}(\rho)\log \Hus_J^{(\lambda)}(\rho) \domega +\; \mathcal{E},
\end{align}
with $\abs{\mathcal{E}}\leq C\log(2K+1)\tfrac{(2J+1)^2}{2K+1}$.
\end{proposition}

\begin{proof}
    We proceed as in the proof of \Cref{lemma:OpJ_functional_calculus_specialized_to_entropy} and apply \Cref{thm:c2_and_holder_convergence_of_traces_channels} to get:
    \begin{align}
        \Bigg\vert\frac{1}{2K+1}\Tr_K\Bigg(\varphi\Bigg(\frac{2K+1}{2J+1}&\Phi_{J,K}^{K+i}(\rho)\Bigg)\Bigg)-\int_{S^2}\varphi(\Hus_J^{-i}(\rho))(\omega)\domega\Bigg\vert\\
        &\leq \frac{C}{\alpha(1-\alpha)} \left(\frac{J-\abs{i}+1}{2K-J+i+1}\right)^\alpha + C \frac{J-\abs{i}+1}{2K-J+i+1}.\nonumber
    \end{align}
    We choose $\alpha = 1-1/\log(2K+1)$, which is a valid H\"older exponent as long as $K\geq 1$. We now bound:
    \begin{align}
        \left(\!\frac{J-\abs{i}+1}{2K-J+i+1}\!\right)^{-\frac{1}{\log(2K+1)}} \!\!\!\!\!= \exp\!\left(\frac{\log(2K-J+i+1)}{\log(2K+1)}-\frac{\log(J-\abs{i}+1)}{\log(2K+1)}\right)\leq e.
    \end{align}
    We are left with:
    \begin{align}
        \Bigg\vert\frac{1}{2K+1}\Tr_K\Bigg(\varphi\Bigg(\frac{2K+1}{2J+1}&\Phi_{J,K}^{K+i}(\rho)\Bigg)\Bigg)-\int_{S^2}\varphi(\Hus_J^{-i}(\rho))(\omega)\domega\Bigg\vert\\&\leq C\log(2K+1)\frac{J-\abs{i}+1}{2K-J+i+1}.\nonumber
    \end{align}
    Finally we use the fact that $\varphi(\lambda x) = \lambda \varphi(x)+x\varphi(\lambda)$ to write:
    \begin{align}
        \frac{2J+1}{2K+1}\Tr_K\Bigg(\varphi\Bigg(\frac{2K+1}{2J+1}&\Phi_{J,K}^{K+i}(\rho)\Bigg)\Bigg) = \log\left({\frac{2K+1}{2J+1}}\right) - S_{vN}(\Phi_{J,K}^{K+i}(\rho)).
    \end{align}
    The result follows. The case of a general channel $\Phi^{(\lambda)}_{J,K}$ is analogous.
\end{proof}

\section*{Aknowledgements}

We would like to thank Matthias Christandl, Jan Derezi\'nski, S\o ren Fournais, Thomas Fraser, Martin Dam Larsen, Robin Reuvers, Freek Witteveen, and Lasse Wolff for helpful discussions and suggestions. This work was partially supported by the Villum Centre of Excellence for the Mathematics of Quantum Theory (QMATH) with Grant No.10059.

\clearpage
\appendix
\renewcommand{\thesection}{\Alph{section}}
\setcounter{definition}{0}
\renewcommand{\thedefinition}{\Alph{section}\arabic{definition}}
\setcounter{equation}{0}
\renewcommand{\theequation}{\Alph{section}\arabic{equation}}
\section{A brief review of quantum channels}\label{Appendix:QC}
We review in this appendix some well known facts about quantum channels. 

\begin{definition}
A quantum channel is a linear map $\Phi\colon B(\mathcal{H})\to B(\mathcal{K})$, where $\mathcal{H}, \mathcal{K}$ are finite dimensional Hilbert spaces, which is:
\begin{enumerate}[label = -]
    \item Completely positive (CP), that is $\Phi\tensor \identity_n\colon B(\mathcal{H})\tensor B(\C^n)\to B(\mathcal{K})\tensor B(\C^n)$ is a~positive map for all $n$;
    \item Trace preserving (TP), that is $\Tr_\mathcal{K}(\Phi(\rho))=\Tr_\mathcal{H}(\rho)$ for all $\rho\in B(\mathcal{H})$.
\end{enumerate}

Let $G$ be a group acting on $\mathcal H, \mathcal K$ by unitary representations $\pi_{\mathcal H}, \pi_{\mathcal K}$. We say that a~channel $\Phi \colon B(\mathcal H) \to B(\mathcal K)$ is $G$-equivariant if 
\begin{align}
   \Phi(\pi_\mathcal H(g) \rho\, \pi_\mathcal H(g)^*) = \pi_\mathcal K(g)\Phi (\rho)\pi_\mathcal K(g)^*\quad\text{for all }g\in G, \ \rho\in B(\mathcal H).
\end{align}
\end{definition}

An important tool in the study of quantum channels is Choi's Theorem. A quantum channel is an element of:
\begin{equation}
 B(B(\mathcal H),B(\mathcal K)) \cong B(\mathcal K)\tensor B(\mathcal H)^*\cong \mathcal K\tensor \mathcal K^*\tensor \mathcal H^*\tensor \mathcal H.   
\end{equation}
Swapping the two middle two tensor factors and regrouping again gives the isomorphism:
\begin{equation}
    B(B(\mathcal H), B(\mathcal K)) \cong B(\mathcal K \otimes \mathcal H^*).
\end{equation}
In more explicit terms, the map $\Phi \in B(B(\mathcal H),B(\mathcal K))$ and the corresponding element $C_\Phi \in B(\mathcal K \otimes \mathcal H^*)$, often called the Choi matrix, are related by the formula:
\begin{align}
    \Phi(\rho) =\Tr_{\mathcal H^*}[C_\Phi (\mathbbm{1}_{\mathcal K} \tensor \rho ^t)],
    \label{eq:channel_from_Choi}
\end{align}
where $\Tr_{\mathcal H^*}$ is the partial trace $B(\mathcal K \otimes \mathcal H^*) \to B(\mathcal K)$ and $\rho^t\in B(\mathcal H^*)$ is the transpose of $\rho\in B(\mathcal H)$:
\begin{equation}
    (\rho^t \phi) (x) = \phi(\rho x) \qquad \text{for } \phi \in \mathcal H^*, \ x \in \mathcal H.
\end{equation}


\begin{theorem}[Choi]
A map $\Phi\colon B(\mathcal H)\to B(\mathcal K)$ is completely positive if and only if its Choi matrix $C_\Phi \in B(\mathcal K \tensor \mathcal H^*)$ is a positive operator. 
\end{theorem}

We remark that the Choi isomorphism is often introduced in a (basis dependent) way which does not use dual spaces. The isomorphism $B(B(\mathcal H),B(\mathcal K)) \cong B(\mathcal K \otimes \mathcal H^*)$ we choose to work with is canonical and, what's important for our purposes, equivariant. Therefore, a map $B(\mathcal H)\to B(\mathcal K)$ is equivariant if and only if its Choi matrix commutes with the action of $G$ on $\mathcal K\tensor \mathcal H^*$, i.e. with the operators $\pi_{\mathcal K}(g) \otimes \pi_{\mathcal H^*}(g)$, where:
\begin{equation}
    \pi_{ \mathcal H^*}(g) =  \pi_{\mathcal H}(g^{-1})^t.
    \label{eq:contragredient}
\end{equation}

Let $\mathcal H$ be an irreducible representation. For a $G$-equivariant map $\Phi\colon B(\mathcal H)\to B(\mathcal K)$, the composition $\Tr_{\mathcal K}\circ \Phi$ is an invariant linear functional on $B(\mathcal H)$. By Schur's lemma it must then be a scalar multiple of $\Tr_{\mathcal H}$, $\Tr_{\mathcal K}\circ \Phi = \lambda \Tr_{\mathcal H}$. That is, every $G$-equivariant map $B(\mathcal H)\to B(\mathcal K)$ is trace preserving up to a scalar factor. One can check that $\lambda=1$ if and only if $\Tr_{\mathcal K\tensor \mathcal H^*}C_\Phi = \dim (\mathcal H)$. Similarly, irreducibility of $\mathcal K$ implies that every equivariant map $B(\mathcal H) \to B(\mathcal K)$ is unital up to a scalar factor. We note that only if $\dim(\mathcal H)=\dim(\mathcal K)$ can a map be trace and preserving and unital (not just up to a~factor) at the same time. 

We summarize this discussion in the following lemma.

\begin{lemma}
Let $\mathcal H$ be an irreducible representation. A map $\Phi\colon B(\mathcal H)\to B(\mathcal K)$ is a~$G$-equivariant quantum channel if and only if its Choi matrix $C_\Phi \in B(\mathcal K \tensor \mathcal H^*)$ is a positive operator which commutes with the action of $G$ and has trace $\dim (\mathcal H)$.
\end{lemma}

This leads to the following characterization of the $G$-equivariant quantum channels between irreducible representations of $G$ in the case that $G$ is a compact Lie group (in~which case all representations are assumed to be continuous).

\begin{corollary} \label{cor:extreme_channels}
Let $\mathcal H$ be an irreducible representation of $G$. A map $\Phi$ is an extreme point of the set of $G$-equivariant quantum channels $B(\mathcal H) \to B(\mathcal K)$ if and only if its Choi matrix $C_\Phi$ is the projection onto an irreducible subrepresentation of $ \mathcal K \tensor \mathcal H^*$.
\end{corollary}

Consider the decomposition of the tensor product representation $\mathcal K \tensor \mathcal H^*$ into isotypic components:
\begin{align}
    \mathcal K \tensor \mathcal H^* \cong \bigoplus_\lambda \mathcal G_\lambda^{\oplus m_\lambda} \cong \bigoplus_\lambda \C^{m_\lambda}\tensor \mathcal G_\lambda,
\end{align}
where $\mathcal G_\lambda$ are the irreducible representations of $G$. By \Cref{cor:extreme_channels}, the set of extreme $G$-equivariant channels $B(\mathcal H) \to B(\mathcal K)$ is homeomorphic to the disjoint union of complex projective spaces $\bigsqcup_\lambda \mathbb{CP}^{m_\lambda-1}$.

\section{Comparison with Berezin-Toeplitz quantization} \label{app:BT}
We can interpret the sphere $S^2$ as the projective space of $\irrep_{1/2}$. As such, it carries a~distinguished hermitian holomorphic line bundle $\ell$, called the tautological bundle. The fiber of $\ell$ over $\omega \in S^2$ is the linear span of $\ket{\omega}_{1/2}$. There is an isomorphism of $\mathcal H_J$ with the symmetric $2J$-fold tensor product of $\irrep_J$,
\begin{equation}
\irrep_J \cong \mathrm{Sym}^{2J} \mathcal H_{1/2},
\end{equation}
taking $\ket{\omega}_J$ to $\ket{\omega}_{1/2}^{\otimes 2J}$. Therefore, the fiber over $\omega$ of the line bundle $\ell^{2J}$ (the $2J$-fold fiberwise tensor product of $\ell$) may be identified with the linear span of $\ket{\omega}_J$. The bundles $\ell^{2J}$ with $J>0$ do not admit any globally defined holomorphic sections. Instead we consider $\ell^{-2J}$, the dual bundle of $\ell^{2J}$. We remark that $\ell^{-1}$ is often called the hyperplane line bundle. Let $\mathcal O(S^2, \ell^{-2J})$ be the space of global holomorphic sections of $\ell^{-2J}$. There exists a unitary isomorphism $V : \irrep_J \to \mathcal O(S^2, \ell^{-2J})$ given as follows:
\begin{equation}
    V \ket{\psi} (\omega) = \sqrt{2J+1} \left. \beta(\ket{\psi}, \cdot) \right|_{\mathbb C \ket{\omega}_J},
\end{equation}
where $\beta$ is the (essentially unique) invariant bilinear form on $\mathcal H_J$, see \eqref{eq:beta_def}. Since $\mathcal O(S^2, \ell^{-2J})$ is finite-dimensional, it is in particular a closed subspace of the Hilbert space of square-integrable sections of $\ell^{-2J}$. Let $P$ be the orthogonal projection onto $\mathcal O(S^2, \ell^{-2J})$. If $f$ is a function on $S^2$, the associated Toeplitz operator on $\mathcal O(S^2, \ell^{-2J})$ is defined as $T_f=PfP$, i.e.\ as the compression of $f$ regarded as a multiplication operator. One can check that
\begin{equation}
    \Op_J(f) = V^{-1} T_{f \circ a} V ,
\end{equation}
where $a : S^2 \to S^2$ is the antipodal map, $a (\omega) = - \omega$.

\printbibliography
\end{document}